%% file: main.tex
\documentclass[11pt,letterpaper]{article}
\usepackage{amsmath,amsfonts,amsthm,amssymb}
\usepackage{setspace}
\usepackage{fancyhdr}
\usepackage{lastpage}
\usepackage{extramarks}
\usepackage{xspace}
\usepackage{chngpage}

\usepackage{makecell}
\usepackage{enumitem} 
\usepackage{multirow}

\usepackage{algorithm,algorithmicx}
\usepackage[noend]{algpseudocode}
\usepackage{graphicx}
\usepackage{nicefrac,bbm}
\usepackage{subcaption}
\usepackage{thm-restate}
\usepackage{tikz}
\usetikzlibrary{calc}

\usepackage{graphicx}

\usepackage{soul,color}
\usepackage{graphicx,float,wrapfig}
\usepackage[font=small,labelfont=bf]{caption}
\usepackage[margin=1in]{geometry}
\usepackage{enumitem}
\usepackage{thmtools,thm-restate}
\usepackage{bbm}
\usepackage{parskip}

\allowdisplaybreaks

\usepackage{mdframed}
\newtheorem*{mdresult}{Result}

\usepackage[dvipsnames]{xcolor}
\usepackage[linktocpage=true,breaklinks,colorlinks,citecolor=blue,linkcolor=BrickRed,psdextra]{hyperref}
\usepackage{cleveref}

\usepackage{wrapfig}
  
\title{\textmd{\bf Asymmetric Trading Prophets
  }}

\date{\today}
\author{
Gagan Aggarwal\thanks{Google Research (\texttt{gagana@google.com}). }
\and
Anupam Gupta\thanks{Department of Computer Science, New York
  University. (\texttt{anupam.g@nyu.edu}) Work supported in part by NSF awards CCF-2422926 and CCF-2608359.}
\and
Yifan Wang\thanks{School of Computer Science, Georgia Tech (\texttt{ywang3782@gatech.edu}). Part of this work was done while the author was visiting Google as a Student Researcher.}
\and
Mingfei Zhao\thanks{Google Research (\texttt{mingfei@google.com}).}
}

\input{macros}

\begin{document}

\maketitle 

\begin{abstract}


  The ``Trading Prophet'' problem challenges an online trader to
  maximize its profit by buying and selling assets under stochastic
  prices and capacity constraints, competing against an offline
  prophet with full foresight. In previous work, each arriving asset
  was assumed to have a single price $p_t$, and the trader was allowed
  to either buy a copy at this price (subject to having available
  capacity), or sell a copy (if it already held at least one copy in
  hand). However, this abstraction can fail to capture the structural
  asymmetry of decentralized dealer-based markets, where buying and
  selling opportunities could be distinct, and driven by individual
  preferences. To address this, we introduce the \emph{Asymmetric
    Trading Prophets} problem, where at each timestep the trader
  observes a price tuple $(b_t, s_t)$---representing the cost to buy,
  and the revenue from selling at this timestep. Importantly, the
  $(b_t,s_t)$ tuple could be potentially arbitrarily correlated.

  We provide the first competitive analysis for this asymmetric
  trading prophets problem, characterizing the achievable profit based
  on the trader's capacity $B$ and initial inventory $B_0$. For the
  unit-capacity case of $B=1$,
  we design online algorithms that achieve constant competitive ratios
  for both i.i.d.\ and non-i.i.d.\ distributions on the price tuples,
  when the trader has one initial copy ($B_0=1$). For the general
  capacity case where $B$ can be large, we give algorithms for i.i.d.\
  distributions that achieve a competitive ratio of
  $1 - \Theta(\log B_0/\sqrt{B_0})$. Finally, for the
  \emph{symmetric case} (where the price tuple satisfies $b_t=s_t$),
  we improve this to get a competitive ratio of
  $1 - O(\log B/\sqrt{B})$, demonstrating that the
  performance approaches optimality as the capacity increases. We show
  that both ratios are tight up to a logarithmic factor.


\end{abstract}






\clearpage

\input{intro}


\input{single-noniid}

\input{iid}

\input{lower}
{\small
\bibliographystyle{alpha}
\bibliography{ref.bib}
}

\appendix

\input{ineq}

\input{single-noniid-missing}

\input{iid-mainalg-missing}


\end{document}

%% file: macros.tex
\DeclareUnicodeCharacter{2217}{*}

\newtheorem{Theorem}{Theorem}[section]
\newtheorem{Lemma}[Theorem]{Lemma}

\newtheorem{Definition}[Theorem]{Definition}

\newtheorem{Claim}[Theorem]{Claim}
\newtheorem{Corollary}[Theorem]{Corollary}

\newcommand{\parta}{(\text{\uppercase\expandafter{\romannumeral1}})}
\newcommand{\partb}{(\text{\uppercase\expandafter{\romannumeral2}})}
\newcommand{\partc}{(\text{\uppercase\expandafter{\romannumeral3}})}

\newcommand{\E}{\mathbb{E}}

\newcommand{\one}{\mathbf{1}}%
\newcommand{\bluemath}[1]{{\color{blue}{#1}}}
\newcommand{\pr}{\mathbf{Pr}} 

\newcommand{\ignore}[1]{{}}

\newcommand{\D}{\mathcal{D}}
\newcommand{\cF}{\mathcal{F}}
\newcommand{\cE}{\mathcal{E}}

\newcommand{\IGNORE}[1]{}
\newcounter{note}[section]

\newcommand{\calI}{\mathcal{I}}

\newcommand{\lp}{{\mathsf{LP}}}

\newcommand{\nf}{\nicefrac}

\newcommand{\GTP}{{\textsf{ATP}}\xspace}
\newcommand{\ATP}{\GTP}

\newcommand{\var}{\mathsf{Var}}


%% file: intro.tex
\section{Introduction}


The maxim ``buy low, sell high'' captures the fundamental goal of any
trader, yet executing this strategy in an online setting involves
significant uncertainty. This challenge is formally captured by the
\emph{Trading Prophets} problem, introduced by \cite{CCDHOS-EC23}. In
this setting, a trader observes a sequence of stochastic prices for
the same asset; they must decide, irrevocably and in real-time,
whether to buy a copy of the asset (if they have available capacity)
or to sell a copy (if they have at least one copy they can sell), to maximize
profit within storage constraints. The objective is to design online
algorithms that compete against the ``prophet''---an offline
benchmark with full foresight of future price realizations. This
framework bridges the classic theory of prophet inequalities and
optimal stopping with the dynamics of inventory management.

The work by \cite{CCDHOS-EC23} and \cite{azar2026trading} has
established strong theoretical foundations for this problem. These
works model the market as an exogenous ticker tape: at any time step
$t$, there is a single ``market price'' drawn from a known
distribution. The trader chooses to transact (either buy or sell) at
this price. This abstraction effectively models centralized exchanges,
such as the stock market, where a unified price dictates terms for all
participants and the cost to buy is tightly coupled to the revenue
from selling. 

However, this ``single price'' abstraction fails to capture the
structural asymmetry inherent in other decentralized or dealer-based
markets. In these settings, a trader acts as a market-maker, facing a
stream of distinct agents, each with their own private valuations. For
instance, in the market for used vehicles, the price at which a dealer
can acquire inventory from a desperate seller can be often entirely
distinct from the price a collector is willing to pay for that same
item, even if they arrive at the same time. Here, the ``price'' is not
a single scalar, but rather a set of distinct buying and selling
opportunities driven by individual preferences rather than a unified
market signal. 

To address this gap, we introduce the \emph{Asymmetric Trading Prophet
  (\ATP)} problem. In our model, at each time step $t$, the trader
observes a price tuple $(b_t, s_t)$ drawn from a known distribution,
where $b_t$ represents the cost to purchase an item and $s_t$
represents the revenue from selling one. Crucially, we allow $b_t$ and
$s_t$ to be arbitrarily correlated.
The trader must now perform exactly one of the following actions: buy,
sell, or skip. This framework strictly generalizes the ``symmetric''
stock market model (by setting $b_t = s_t$) while capturing scenarios
inaccessible to prior works, such as disjoint buy-only or sell-only
opportunities (e.g., via $s_t = 0$ or $b_t = \infty$), and general
asymmetric price structures. Having formulated this model, we ask:

\begin{quote}
\emph{Question 1: Can we design online algorithms that achieve (expected) profit comparable to the offline optimal prophet in hindsight in the Asymmetric Trading Prophet problem?}
\end{quote}

The original Trading Prophet model proposed by \cite{CCDHOS-EC23} and
\cite{azar2026trading} assumes that the trader has
\emph{unit-capacity}, and 
can hold at most one item.\footnote{\cite{CCDHOS-EC23} consider
  settings where the trader can store up to $k$ items, but allow the
  trader to buy or sell the entire inventory at a single time step,
  which makes this variant equivalent to the single item case; see
  \cite[\S5.2]{CCDHOS-EC23}.} While this unit-capacity
assumption is sufficient for stock markets where inventory can be
liquidated instantly, it is restrictive for markets like used
vehicles. In such dealer-based settings, buyers and sellers are
individuals who arrive sequentially, meaning the trader can transact
at most one copy of the item per time step. Consequently, the dealer
essentially functions as an inventory buffer, requiring a capacity
$B \ge 1$ to hold multiple copies of the same good. This allows the
trader to accumulate stock during periods of low-cost supply and
gradually release it as high-value buyers arrive. In this paper, we
incorporate this constraint by allowing the trader to hold up to $B$
copies of the item, while limiting trades to one copy per time
step. We study the competitive ratio in terms of the trader’s capacity
and initial inventory. Even in the special unit-capacity stock market
setting, the competitive ratio is unbounded if the trader has no
copies initially \cite{CCDHOS-EC23} yet becomes constant when it has
one initial copy \cite{azar2026trading}. We ask:
\begin{quote}
  \emph{Question 2: How does the capacity and initial inventory of the
    trader affect the competitive ratio of the online algorithms?}
\end{quote}

\subsection{Our Model}
\label{sec:our-model}

In the \emph{Asymmetric Trading Prophets} (\ATP) problem, a decision
maker interacts with $T$ sequentially arriving requests to trade
copies of a certain product. The decision maker has a capacity
constraint of $B \geq 1$, which means she can hold at most $B$ copies
of the product at any time. Before the decision making process begins,
the decision maker has a initial inventory of $B_0 \in [0, B]$
copies. We assume $T \geq B_0$, while $B \leq B_0 + T$, otherwise
either the initial inventory level $B_0$, or the capacity $B$ would
become unnecessary. (This is only for simplicity; we can discharge
these assumptions by replacing occurrences of $B_0$ by $\min(B_0,T)$,
and those of $B$ by $\min(B_0+T, B)$ in our results.)

The request $t$ consists of a price tuple $(b_t, s_t)$ presented to
the decision maker, where $(b_t, s_t)$ is independently drawn from a
publicly known distribution $\D_t$. In general, we allow the price
tuple to be \emph{asymmetric}, i.e., the prices $b_t$ and $s_t$ are
not necessarily equal. (Some of our results are specialized to the
\emph{symmetric} setting, where $b_t=s_t$ for any tuple in the support
of $\D_t$.)  Faced with the request, the decision maker takes exactly
one of the following three actions:
\begin{itemize}
\item \textbf{Buy}: if the decision maker has at most $B - 1$ copies of the product, she can buy a new copy at price $b_t$.
\item \textbf{Sell}: if the decision maker has at least $1$ copy of the product, she can sell a copy at price $s_t$.
\item \textbf{Skip}: the decision maker does nothing.
\end{itemize}

The goal of the decision maker is to maximize their total profit,
i.e., the total revenue obtained from sales, minus the total cost
incurred from purchases. In this paper, we evaluate the performance of
an online algorithm for the \ATP problem via the framework of
\emph{competitive analysis}. i.e., we seek to design good online
algorithms such that the expected profit achieved by the online
algorithm is at least factor times the optimal offline algorithm,
i.e., the optimal achievable profit assuming knowledge of all the
arriving price tuples.

\subsection{Main Results}
\label{sec:main-results}

In this paper, we study the \ATP problem and give answers to both questions mentioned above. Our results are summarized in \Cref{tab:results}.

\begin{table}[tbh]
\centering
\renewcommand{\arraystretch}{1.6}
\resizebox{\textwidth}{!}{%
    \begin{tabular}{c|cc|cc}
    \hline
    & \multicolumn{2}{c|}{\textbf{Unit-Capacity}}
    & \multicolumn{2}{c}{\textbf{Large Capacity}} \\
    \cline{2-5}
    & \textbf{Upper Bound} & \textbf{Lower Bound}
    & \textbf{Upper Bound} & \textbf{Lower Bound} \\
    \hline

    Non-I.I.D., General
    & \makecell{$\bluemath{22}$ (Thm.~\ref{thm:label-noniid-small-informal})}
    & {\makecell{$3$ [A+'26]}}
    & -
    & -
    \\

    Non-I.I.D., Symmetric
    & \makecell{$3$ [A+'26]}
    & {\makecell{$3$ [A+'26]}}
    & -
    & -
    \\

    I.I.D., General
    & \makecell{\bluemath{$4$} (Thm.~\ref{thm:label-iid-small-informal})}
    & \makecell{$1.34$ [C+'17]}
    & \makecell{$\bluemath{1 + O(\log B_0/\sqrt{B_0})}$
      (Thm.~\ref{thm:label-iid-large-informal})} &
            \makecell{$\bluemath{1 + \Omega(1/\sqrt{B_0})}$ (Thm.~\ref{thm:label-iid-large-informal})}
    \\

    I.I.D., Symmetric
    & \makecell{$2$ [C+'23]}
    & - 
    & \makecell{$\bluemath{1 + O(\log B/\sqrt{B})}$ (Thm.~\ref{thm:unlabel-iid-large-informal})} & \makecell{$\bluemath{1 + \Omega(1/\sqrt{B})}$ (Thm.~\ref{thm:unlabel-iid-large-informal})}
    \\

    \hline
    \end{tabular}%
}
\caption{The achievable competitive ratios of the asymmetric trading
  prophet problem. Here, [A+'26] is the work of \cite{azar2026trading}, [C+'17] is
  \cite{CFHOV-EC17}, and [C+'23] is \cite{CCDHOS-EC23}. Results proved in this paper are in blue. All results assume $B_0 \geq 1$.}
\label{tab:results}
\end{table}







\vspace{0.5em} \noindent \textbf{Result for the no initial inventory
  case.} Our first result is for the \ATP problem with no initial
inventory (i.e., with $B_0 = 0$). In contrast to the symmetric setting
where a constant competitive ratio is achievable when the request
distributions are identical \cite[Theorem 1]{CCDHOS-EC23}, we show
that no finite competitive ratio is achievable for the asymmetric case
when $B_0 = 0$, even in the i.i.d.\ setting.

\begin{Theorem}[No Initial Inventory]
  \label{thm:no-initial}
  For any maximum capacity $B$, there exists an \ATP instance with
  $B_0 = 0$ such that the optimal offline algorithm achieves an
  expected profit strictly greater than $0$, while no online algorithm
  can achieve an expected profit better than $0$.
\end{Theorem}

Given \Cref{thm:no-initial}, we now focus on cases
when the initial inventory $B_0$ is at least $1$.

\vspace{0.5em} \noindent \textbf{Results for the unit-capacity case.}
Our first set of positive results are for the unit-capacity case,
where the decision-maker only has capacity $B =1$; i.e., it can only
hold at most one copy in its inventory at any time. For the general
case where the requests distributions are non identical, we show the
existence of a $22$-competitive algorithm.

\begin{Theorem}[non-I.I.D.\ Unit-Capacity]
  \label{thm:label-noniid-small-informal}
  For \ATP problem with non-identical request distributions, when
  $B_0 = B = 1$, there exists an online algorithm that achieves an
  expected profit of at least $\nf{1}{22}$ times the expected profit
  of the optimal offline algorithm.
\end{Theorem}

Compared to the \ATP problem with symmetric prices studied in
\cite{azar2026trading}, an immediate challenge in the asymmetric-price
setting lies in the complexity of the optimal offline
algorithm. Unlike the symmetric case, where the optimal offline
decision depends only on the current price and the next arriving
price, in the asymmetric case the optimal decision may depend on the
entire remaining suffix of arrivals. 
To see this, consider a simple instance where the trader with one unit of initial inventory and unit capacity facing a sequence where $s_1=10$ and $s_T=100$. In the asymmetric setting, the optimal decision at $t=1$ depends entirely on whether the inventory can be replenished later to capture the final high value. If there exists an intermediate time $k \in (1, T)$ with a low buying price (e.g., $b_k=2$), the optimal strategy is to sell at $t=1$, restock at $t=k$, and sell again at $t=T$. However, if no such buying opportunity exists (i.e., $b_t = \infty$ for all $1 < t < T$), selling at $t=1$ would be a catastrophic error. This dependency prevents us from
providing a direct analysis.

To address this issue, we propose a linear programming relaxation of
the optimal offline algorithm. This relaxation captures the
probability that the optimal offline algorithm buys a copy from
request $i$ and subsequently sells it to request $j$, allowing us to
interpret the optimal LP solution as a collection of fractional
trades.

To further round the optimal LP solution, our core idea is to
decompose it into multiple intervals such that only a constant
fraction of cross-interval trades is discarded, while each interval
contains only a constant number of fractional trades. We then propose
an algorithm using the Online Contention Resolution Scheme (OCRS) \cite{FSZ-SODA16} that extracts a constant fraction of the
profit from each interval that contains only a constant number of
trades.

Next, for the unit-capacity case, we further improve the competitive
ratio to $4$ when the request distributions are identical.

\begin{Theorem}[I.I.D.\ Unit-Capacity]
  \label{thm:label-iid-small-informal}
  For the \ATP problem with identical request distributions, when
  $B_0 = B = 1$, there exists an online algorithm that achieves an
  expected profit of at least $\nf{1}{4}$ times the expected profit of
  the optimal offline algorithm.
\end{Theorem}

To make use of the assumption of identical distributions, our main
idea is to compare against the optimal offline algorithm which also
allows to exchange the arrival order of requests. We then propose a
new linear programming relaxation for such an optimal algorithm, where
the LP variables depend only on the support of the identical request
distribution and are independent of the arrival order of the
requests. Finally, we provide a rounding algorithm to obtain a
$\nf{1}{4}$-fraction of the optimal LP value.


\vspace{0.5em} \noindent \textbf{Results for the large capacity case.}
We then turn to the setting with larger capacity $B$ and initial
inventory $B_0$: for this setting, we focus on the i.i.d.\ case, where
the price tuples at each timestep are drawn from the same
distribution. For this setting, we show that the competitiveness
mainly depends on the initial inventory $B_0$ when price tuples are
(possibly) asymmetric.

\begin{Theorem}[Asymmetric Large-Capacity]
  \label{thm:label-iid-large-informal}
  For the \ATP problem with identical request distributions and $B_0 \geq
  1$, we show the following results:
  \begin{enumerate}[nosep,label=(\alph*)]
  \item \label{item:asym-large-ub} There exists an online algorithm that achieves an expected
    profit of at least $1 - O(\log B_0/\sqrt{B_0})$ times the expected
    profit of the optimal algorithm. 
  \item \label{item:asym-large-lb} On the other hand, no online
    algorithm can obtain a competitive ratio better than
    $1 - \Omega(1/\sqrt{B_0})$.
  \end{enumerate}
\end{Theorem}

We also show that the dependency on $B_0$ in
\Cref{thm:label-iid-large-informal} can be further improved to the
maximum capacity $B$ when the price tuples are symmetric.

\begin{Theorem}[Symmetric Large-Capacity]
  \label{thm:unlabel-iid-large-informal}
  For the \ATP problem with identical request distributions (with
  \emph{symmetric} prices) and $B_0 \geq
  1$, we show the following results:
  \begin{enumerate}[nosep,label=(\alph*)]
  \item \label{item:sym-large-ub} There exists an online algorithm that achieves an expected
    profit of at least $1 - O(\log B/\sqrt{B})$ times the expected
    profit of the optimal algorithm. 
  \item \label{item:sym-large-lb} On the other hand, no online
    algorithm can obtain a competitive ratio better than
    $1 - \Omega(1/\sqrt{B})$.
  \end{enumerate}
\end{Theorem}

To prove the algorithmic results of
\Cref{thm:label-iid-large-informal} and
\Cref{thm:unlabel-iid-large-informal}, we adopt the same LP relaxation
as in \Cref{thm:label-iid-small-informal}. Given an optimal LP
solution, we design a unified rounding algorithm that applies to both
the general setting and the symmetric setting. The rounding procedure
is based on a standard independent rounding scheme, augmented with an
additional mechanism that ensures nearly all inventory is cleared by
the end of the trading process.

\subsection{Related Works}

\vspace{0.5em} \noindent \textbf{Online Trading and Trading Prophets.}
Our work builds directly upon the recent framework of \textit{Trading
  Prophets}, introduced by \cite{CCDHOS-EC23}. They initiated the
study of an online trader who buys and sells a single asset to
maximize profit against an offline prophet, establishing a tight
2-competitive ratio for i.i.d.\ prices and investigating the random
order model. Subsequent work has expanded this model in several
directions: \cite{rajput2025trading} generalized the setting to
multiple stocks and matroid constraints, while \cite{azar2026trading}
showed that while competitiveness is impossible with zero initial
capital, constant competitive ratios are achievable if the trader is
endowed with initial inventory. Our work strictly generalizes these
results by introducing the Asymmetric Trading Prophet problem, which
captures distinct buy and sell prices. In the non-stochastic setting,
\cite{azar2025competitive} studied a fully adversarial model for
bundle trading with bid-ask spreads, and obtained logarithmic
competitive ratios with resource augmentation.

\vspace{0.5em} \noindent \textbf{Prophet Inequalities and Optimal Stopping.}
The \ATP problem is rooted in the extensive literature on prophet
inequalities. The classic result by \cite{krengel1978semimartingales}
established that a single threshold achieves a $\nf12$-approximation
to the prophet's expected maximum, a ratio shown to be tight by
\cite{samuel1984comparison}. Given the vast literature that has
followed, starting from~\cite{hajiaghayi2007automated,chawla2010multi,Alaei-SICOMP14},
we refer the reader to the survey by \cite{correa2019recent} and the
textbook by \cite{roughgarden2021beyond} for a comprehensive
overview. While classic prophet inequalities focus on one-way
selection (stopping to maximize value), our work—and the broader
trading prophet framework—incorporates the dynamic inventory
management aspect, where decisions are reversible and the objective
allows for mixed-sign rewards. Unlike standard multi-unit prophet
inequalities where capacity is static, our capacity constraint $B$
acts as a buffer for temporary inventory, creating a dependency
structure on the entire suffix of arrivals that is absent in the
classic setting.

\vspace{0.5em} \noindent \textbf{Online Contention Resolution Schemes.}
Our rounding techniques rely on the framework of Online Contention
Resolution Schemes (OCRS), formally introduced by \cite{FSZ-SODA16}
to capture the loss inherent in converting fractional solutions to
integral online policies, extending the offline contention resolution
schemes from \cite{ChekuriVZ14}. This framework has become a standard
tool in online stochastic optimization, with applications ranging from
submodular maximization~\cite{lee2018online} to stochastic
matching~\cite{adamczyk2018improved}. Recent work has further expanded
the settings of prophet inequalities and online contention resolution
to handle correlation constraints, e.g.,
\cite{immorlica2020prophet,CaragiannisGLW21,ChekuriL21,DughmiKP24,gupta2024pairwise,LivanosP024,MaurasMR24,FeldmanMMR25}.


\subsection{Paper Organization}

In \Cref{sec:unit-noniid}, we present a constant-competitive algorithm
for the \ATP problem with unit capacity and \emph{non-identical} requests,
and prove \Cref{thm:label-noniid-small-informal}. In
\Cref{sec:iid-upper}, we study the \ATP problem with \emph{identical} request
distributions and provide proofs of
\Cref{thm:label-iid-small-informal} as well as the positive results
of 
\Cref{thm:label-iid-large-informal} and \Cref{thm:unlabel-iid-large-informal}. Finally,
in \Cref{sec:lower}, we establish hardness results for the \ATP
problem and provide proofs of \Cref{thm:no-initial}, as well as the
negative results from
\Cref{thm:label-iid-large-informal} and \Cref{thm:unlabel-iid-large-informal}.


%% file: single-noniid.tex
\section{Non-Identical Requests with Unit Capacity}
\label{sec:unit-noniid}

In this section, we analyze the unit-capacity case ($B = 1$) with
non-identical request distributions, and prove
\Cref{thm:label-noniid-small-informal}. Recall that when we have no
initial copy, no competitiveness is
achievable even when distributions are identical; given this,
we focus on the case when $B_0 = 1$.

\vspace{0.5em} \noindent \textbf{Bernoulli Assumption.} For simplicity of analysis, in this
section, we only present algorithms with the following simplifying
assumption: the support size of each request distribution $\D_t$ is
$2$, such that with probability $p_t$, the arriving price tuple is
$(b_t, s_t)$, where the decision maker knows parameters
$p_t, b_t, s_t$ in advance; with probability $1 - p_t$, the arriving
price tuple is $(+\infty, 0)$, i.e., skipping would be the only
available action for the decision maker. Furthermore, to simplify notations, we define $\D_0$
to be a deterministic distribution with $p_0 = 1$, and
$(b_0, s_0) = (0, 0)$, i.e., we use $\D_0$ to represent the assumption
of having $B_0 = 1$ initial inventory.

This simplified Bernoulli case captures most ideas behind our
algorithm, and we provide a discussion in \Cref{sec:unit-noniid-generalize} to show how to
generalize our algorithm for this Bernoulli case to arbitrary request
distributions with finite support.

\subsection{Main Ideas behind Proving \Cref{thm:label-noniid-small-informal}}

We start with giving a high-level overview of our proof. To prove
\Cref{thm:label-noniid-small-informal}, our algorithm consists of the
following three main ingredients: (a) writing a linear-programming
relaxation, (b) decomposition of the time horizon into disjoint
intervals, such that it suffices to ``make a limited number of
trades'' in each interval, and (c) showing how to efficiently
implement the limited number of trades. We now elaborate on each of
these ingredients.

\subsubsection{The Linear-Programming Relaxation}

In stepping away from the symmetric prices case where $b_t = s_t$, an immediate
challenge we face is that the structure of the optimal offline
algorithm is much more complicated. Indeed, when $b_t = s_t$, the
optimal offline optimum is \emph{localized}: by properly designing the
algorithm, the optimal decision in the current step only depends on
the current price, and the price for the previous and next request (see Algorithm 1 in \cite{azar2026trading}). However, in the
general case, the optimal decision in one step may depend on price
tuple arriving far after the current request. This makes it difficult
for us to directly compare to the profit achieved by the optimal
offline algorithm. Given this conceptual difficulty, we introduce the
following LP relaxation:
\begin{alignat}{2}
  \text{maximize} \qquad 
\sum_{i = 0}^{T-1} \sum_{j = i+1}^T & x_{i, j} \cdot (s_j - b_i) && \tag{$\text{LP}_{\textsc{Relax}}$} \label{program:relax} \\
  \text{s.t.}\quad
  \sum_{j = 0}^{i-1} x_{j, i} + \sum_{j = i+1}^T x_{i, j}
    &\leq p_i &\qquad& \forall i \in [T] \cup \{0\},  \label{eq:LP-marginals} \\
    \sum_{i = 0}^{t-1} \sum_{j = t}^T x_{i, j} &\leq 1 && \forall t
                                                          \in
                                                          [T], \label{eq:lp-capacity}
  \\
  x_{ij} &\geq 0. && \notag
\end{alignat}

Intuitively, the LP variable $x_{i,
  j}$ represents the probability that the algorithm buys a copy from
request $i$, and sells that copy to request
$j$. The constraints~(\ref{eq:LP-marginals}) ensure that the total probability of
buying or selling at time $i$ cannot exceed
$p_t$; the constraints~(\ref{eq:lp-capacity}) ensure that the algorithm can
hold at most one item when going from timestep $t-1$ to timestep $t$.

Given a feasible solution $x$ to \ref{program:relax}, for each
timestep $i \in [T] \cup \{0\}$, define
\[
  q^b_i(x) ~:=~ \sum_{j = i+1}^T x_{i,j} \qquad \text{and} \qquad q^s_i(x) ~:=~ \sum_{j = 0}^{i-1} x_{j, i}
\]
to be the probability that we buy from (and sell to, respectively)
request $i$ when following solution $x$.  Observe that
$q^s_0(x) = q^b_T(x) = 0$. Given this notation, the
constraints~(\ref{eq:LP-marginals}) can be compactly rewritten as
\begin{gather}
  \forall i \in [T] \cup \{0\}, q^s_i(x) + q^b_i(x) ~\leq~ p_i. \label{eq:compact-packing}
\end{gather}
If $\lp(x)$ is the objective value of $x$, then the definition of
$q^b_i(x)$ and $q^s_i(x)$ allows us to rewrite this objective value as 
\begin{gather}
  \lp(x) ~:=~ \sum_{i = 0}^{T-1} \sum_{j = i+1}^T x_{i, j} \cdot (s_j
  - b_i) ~=~ \sum_{i = 1}^T q^s_i(x) \cdot s_i - \sum_{i = 0}^{T-1}
  q^b_i(x) \cdot b_i. \label{eq:lp-value}
\end{gather}

\begin{Lemma}[Relaxation]
  \label{lma:lprelax}
  The linear program \ref{program:relax} is a feasible relaxation of
  the single-capacity \GTP problem. In other words, if $x^*$ is the
  optimal solution to \ref{program:relax}, then $\lp(x^*)$ is at least
  the expected profit of the optimal offline algorithm for the \GTP problem.
\end{Lemma}

Given \Cref{lma:lprelax}, it now suffices to provide a online
algorithm whose profit is at least a constant fraction of $\lp(x^*)$.

\subsubsection{Trading Within Intervals}

Our second idea is to decompose a feasible LP solution $x$ into
intervals. At a high-level, this idea is motivated by a key property
of the algorithm in \cite{CCDHOS-EC23} for the special case of
$b_t = s_t$ with i.i.d.\ requests. A key invariant of their
median-price algorithm is this: before and after the arrival of each
request, the probability that the algorithm holds an item in hand is
\emph{exactly} $0.5$. We now ask: \emph{can we obtain a similar
  structure for the general case?}

We show that such a structure is achievable for the general case if
the feasible LP solution $x$ can be properly decomposed into
intervals. Specifically, we introduce the following definition of an
\emph{$m$-trades solution} and an \emph{$m$-trades decomposition}: (In
this section, we consider settings where $m$ is a small constant.)

\begin{Definition}
  \label{def:mtrades}
  Given a feasible solution $x$ for the linear program
  \ref{program:relax}, an interval
  $[\ell, r] \subseteq [1, T]$,\footnote{Observe that $\ell$
    cannot equal $0$ in this definition.} and a value $m \geq 1$, we
  say that the solution $x$ is an \emph{$m$-trades solution} for
  interval $[\ell, r]$ if:
  \begin{enumerate}[label=(\roman*)]
  \item \label{item:mtrades-i} every $(i, j)$ such that $x_{i,j} > 0$ also satisfy $\ell \leq
    i < j \leq r$, and 
  \item \label{item:mtrades-ii} $\sum_{(i, j): \ell \leq i < j \leq r} x_{i, j} \leq m$.
  \end{enumerate}

  Given a feasible solution $x$ for \ref{program:relax}, we say that
  the solution $x$ admits an \emph{$m$-trades decomposition} if we can
  find disjoint intervals $[\ell_1, r_1], \cdots, [\ell_D, r_D]$ that
  satisfy:
  \begin{enumerate}[label=(\alph*)]
  \item \label{item:mtrades-a} for each $d \in [D]$, the induced solution
    $\{x_{i, j} \cdot \one\big[i, j \in [\ell_d, r_d]\big]\}_{0 \leq i
      < j \leq T}$ is an $m$-trades solution on $[\ell_d, r_d]$, and 
  \item \label{item:mtrades-b} $\{[\ell_d, r_d]\}_{d \in [D]}$ is a feasible decomposition of
    $x$, i.e., for every $(i, j): x_{i, j} > 0$, the indices $i$ and
    $j$ must belong to the same interval $[\ell_d, r_d]$.
  \end{enumerate}
\end{Definition}

In other words, a feasible solution $x$ for \ref{program:relax} admits
an $m$-trades decomposition if it can be decomposed into disjoint intervals,
such that the trades represented by $x$ only happen within those
intervals, and at most $m$ trades happen within each interval. Now, to
round the solution $x$, it suffices to consider each interval
independently.
Since the expected number of trades within an interval is small, it
allows us to give an algorithm which extracts a constant fraction of
the expected profit within the interval. Specifically, we have the
following lemma:

\begin{Lemma}[$m$-Trades Lemma]
  \label{lma:one-interval}
  Let $x$ be an $m$-trades solution for interval $[\ell, r]$. Then
  there exists an online algorithm for the \GTP instance, such that:
  \begin{enumerate}[label =(\roman*)]
  \item if the algorithm holds a copy with probability exactly $0.5$
    before request $\ell$ arrives, then the algorithm also holds a
    copy with probability exactly $0.5$ after servicing request
    $r$. Moreover,
  \item the expected profit obtained by the algorithm is exactly
    $\nicefrac{1}{4m} \cdot \lp(x)$.
  \end{enumerate}
\end{Lemma}

As an immediate corollary of \Cref{lma:one-interval}, we get the
following result:

\begin{Corollary}
  \label{cor:multi-intervals}
  Let $x$ be a feasible solution for linear program
  \ref{program:relax}. Suppose $x$ admits an $m$-trades
  decomposition. Then there exists an online algorithm for the \GTP
  instance with an expected profit of exactly $\nicefrac{1}{4m} \cdot \lp(x)$.
\end{Corollary}

\subsubsection{Decomposing the Optimal Solution}

To obtain a constant approximation algorithm, the final component is
to take an optimal solution $x^*$ to the linear program
(\ref{program:relax}) and give a constant-trades decomposition without
reducing the objective function substantially. We do this in two
steps:

\begin{itemize}
\item Step 1: We first convert the optimal solution $x^*$ into a
  ``nicer'' solution $\tilde x$ via \emph{uncrossing} (see \Cref{subsec:proof_decomposition_lemma}), such that
  $\lp(\tilde x) = \lp(x^*)$.
\item Step 2: We then decompose the solution $\tilde x$ into three
  solutions $\tilde x^{(1)}, \tilde x^{(2)}, \tilde x^{(3)}$ which
  satisfies
  \[ \tilde x = \tilde x^{(1)} + \tilde x^{(2)}+ \tilde x^{(3)}, \]
  where $\tilde x^{(1)}$ is a special solution handling trades that
  buy at request $0$. We then show that both $\tilde x^{(2)}$ and
  $\tilde x^{(3)}$ admit a constant-trades decomposition.
\end{itemize}

The following lemma formalizes the above ideas:

\begin{Lemma}[Decomposition Lemma]
  \label{lma:decompose-opt}
  If $x^*$ is an optimal solution for \ref{program:relax}, there
  exists three feasible solutions
  $\tilde x^{(1)}, \tilde x^{(2)}, \tilde x^{(3)}$ for
  \ref{program:relax} such that:
    \begin{enumerate}[label=(\roman*)]
        \item \label{item:decompose-opt-i} $\lp(x^*) = \lp(\tilde x^{(1)}) + \lp(\tilde x^{(2)}) + \lp(\tilde x^{(3)})$,
        \item \label{item:decompose-opt-ii} there exists an online algorithm obtaining an expected
          profit of $\frac{1}{2} \cdot \lp(\tilde x^{(1)})$, and 
        \item \label{item:decompose-opt-iii} both $\tilde x^{(2)}$ and $\tilde x^{(3)}$ admit a $2.5$-trades decomposition.
    \end{enumerate}
\end{Lemma}

We defer the proofs of \Cref{lma:lprelax}, \Cref{lma:one-interval}, and \Cref{cor:multi-intervals} to \Cref{sec:unit-noniid-missing}, and present the proof of \Cref{lma:decompose-opt} in \Cref{subsec:proof_decomposition_lemma}.  Combining these lemmas proves
\Cref{thm:label-noniid-small-informal}, as we now show.

\begin{proof}[Proof of \Cref{thm:label-noniid-small-informal}]
  We first solve \ref{program:relax} and get the optimal solution
  $x^*$. Next, we apply \Cref{lma:decompose-opt} to obtain three
  solutions $\tilde x^{(1)}, \tilde x^{(2)}, \tilde x^{(3)}$. We now
  do the following:
  \begin{itemize}
  \item w.p.\ $\nicefrac{1}{11}$, we run the algorithm
    provided by \Cref{lma:decompose-opt} to 
    obtain a profit of $\frac{1}{2} \cdot \lp(\tilde x^{(1)})$, else 
  \item we run the algorithm provided by \Cref{cor:multi-intervals} on
    one of $\tilde x^{(2)}$ or $\tilde x^{(3)}$, each w.p.\
    $\nicefrac5{11}$.
  \end{itemize}
  Since $\tilde x^{(2)}$ and $\tilde x^{(3)}$ both admit a
  $2.5$-trades decomposition, our final expected profit is
  \begin{align*}
    \frac{1}{11} \cdot \frac{\lp(\tilde x^{(1)})}{2} + \frac{5}{11}
    \cdot \frac{1}{10} \cdot (\lp(\tilde x^{(2)}) + \lp(\tilde
    x^{(3)})) ~=~ \frac{\lp(\tilde x^{(1)}) + \lp(\tilde x^{(2)}) +
    \lp(\tilde x^{(3)})}{22} ~=~ \frac{\lp(x^*)}{22}. 
  \end{align*}
  Since \Cref{lma:lprelax} guarantees that $\lp(x^*)$ is at least the
  expected profit of the optimal offline algorithm, the algorithm
  achieves a $\nicefrac{1}{22}$-approximation.
\end{proof}


\subsection{Proof of the Decomposition Lemma}\label{subsec:proof_decomposition_lemma}

In this subsection, we show how to give a constant-trades decomposition for the optimal solution of \ref{program:relax} and prove \Cref{lma:decompose-opt}. We first note that an arbitrary optimal solution $x^*$ of \ref{program:relax} may be impossible to decompose. For instance, assume $x^*_{1, T} > 0$, i.e., solution $x^*$ includes a trade that buys from request $1$, and sells to request $T$. To include this trade inside the decomposition, it is necessary to include interval $[1, T]$, and this would be the only interval in our decomposition. However, since there can be more than constant number of trades in $[1, T]$, a constant-trades decomposition becomes impossible in this case.

To resolve this issue, we first apply the idea of \emph{uncrossing} to eliminate the existence of such long-length trades in the optimal solution. To be specific, we use \Cref{alg:uncrossing} to come up with a refined optimal solution via uncrossing.

\begin{algorithm}[tbh]
\caption{\textsc{Refining Optimal Solution via Uncrossing}}
\label{alg:uncrossing}
\begin{algorithmic}[1]
\State \textbf{Input:} Optimal solution $x^*$ of \ref{program:relax}.
\State Initialize solution $\tilde x = x^*$, i.e., $\tilde x_{i, j} = x^*_{i, j}$.
\While{there exists $0 \leq i < i' < j' < j$ such that $\tilde x_{i, j}, \tilde x_{i', j'} > 0$}
\State Let $\delta = \min\{\tilde x_{i, j}, \tilde x_{i', j'}\}$.
\State  Decrease $\tilde x_{i, j}, \tilde x_{i', j'}$ by $\delta$, and increase $\tilde x_{i, j'}, \tilde x_{i', j}$ by $\delta$
\EndWhile
\State \textbf{Output:} Refined solution $\tilde x$ 
\end{algorithmic}
\end{algorithm}

Intuitively, \Cref{alg:uncrossing} ensures that the refined solution satisfies a "first buy first sell" policy: Since the bought (fractional) items are interchangable, for each step it sells first at request $j'$ the item bought early at $i$, and then sells the remaining ones at $j$. The refined solution remains feasible and the objective is unchanged, which is formally proved in \Cref{clm:uncrossing}.
\begin{Claim}
\label{clm:uncrossing}
    The output $\tilde x$ of \Cref{alg:uncrossing} is an optimal solution for \ref{program:relax}, i.e., $\tilde x$ is feasible and it satisfies $\lp(\tilde x) = \lp(x^*)$.
\end{Claim}

\begin{proof}
    It suffices to show that after performing one step of uncrossing in \Cref{alg:uncrossing}, the feasibility of $\tilde x$ holds, and the objective remains unchanged. For the objective, note that the objective of \ref{program:relax} only depends on the value of every $q^b_i(\tilde x)$ and $q^s_i(\tilde x)$. Since the values of $q^b_i(\tilde x)$ and $q^s_i(\tilde x)$ remain unchanged after performing one uncrossing step, the objective remains unchanged. For Constraints \eqref{eq:LP-marginals}, it can be rewritten as $q^b_i(\tilde x) + q^s_i(\tilde x) \leq p_i$, which remains feasible as every $q^b_i(\tilde x)$ and $q^s_i(\tilde x)$ remain unchanged. For Constraints \eqref{eq:lp-capacity}, it can be verified that for every $t \in [T]$, the quantity $\sum_{i = 0}^{t-1} \sum_{j = t}^T x_{i, j}$ is unchanged, so the constraint remains feasible.
\end{proof}

Next, we present the algorithm for decomposing the uncrossed optimal solution $\tilde x$. We start from presenting an initial decomposition via \Cref{alg:init-decompose}.

\begin{algorithm}[tbh]
\caption{\textsc{Initial Decomposition for Refined Optimal Solution}}
\label{alg:init-decompose}
\begin{algorithmic}[1]
\State \textbf{Input:} Optimal solution $\tilde x$ of \ref{program:relax} given by \Cref{alg:uncrossing}.
\State Initialize interval set $\calI = \varnothing$.
\State Initialize $\ell = 1$ (note that the algorithm skips request $0$)
\While{$\ell \leq T$}
\State Let $r \in [\ell, T]$ be the smallest index that satisfies
    \[
    \sum_{i = \ell}^r q^b_i(\tilde x) ~\geq~ 1 \qquad \text{and} \qquad \sum_{i = \ell}^r q^s_i(\tilde x) ~\geq~ 1.
    \]
\phantom{xxi} If such index does not exist, set $r = T$.
\State Add $[\ell, r]$ into $\calI$.
\State Update $\ell = r + 1$.
\EndWhile
\State \textbf{Output:} Interval set $\calI$
\end{algorithmic}
\end{algorithm}

\Cref{alg:init-decompose} provides an almost ready decomposition for solution $\tilde x$. To be specific, we show the following \Cref{lma:init-decompose}:

\begin{Lemma}
    \label{lma:init-decompose}
    Assume the output $\calI$ of \Cref{alg:init-decompose} equals to $\{[\ell_1, r_1], \cdots [\ell_D, r_D]\}$, such that intervals are sorted in order, i.e., $r_d + 1 = \ell_{d+1}$ for every $d \in [D - 1]$. Then, the following statements hold when the input $\tilde x$ of \Cref{alg:init-decompose} is provided by \Cref{alg:uncrossing}:
    \begin{enumerate}[label = (\roman*)]
        \item \label{item:init-decompose-i} For every $d \in [D]$, we have $\sum_{i = \ell_d}^{r_d} q^b_i(\tilde x) \leq 2.5$ and $\sum_{i = \ell_d}^{r_d} q^s_i(\tilde x) \leq 2.5$.
        \item \label{item:init-decompose-ii} For every $(i, j)$ that satisfies $\tilde x_{i, j} > 0$ and $i \neq 0$, $i$ and $j$ cross at most two adjacent intervals, i.e., assume $i \in [\ell_d, r_d]$ and $j \in [\ell_{d'}, r_{d'}]$, there must be $d' \leq d + 1$.
    \end{enumerate}
\end{Lemma}

Our proof of \Cref{lma:init-decompose} relies on the following \Cref{clm:bounded-diff}. At a high level, the claim follows the fact that for each interval $[l,r]$, $\sum_{i = \ell}^r q^b_i(x) - \sum_{i = \ell}^r q^s_i(x)$ is the difference between total fractional items after request $r$ and before request $l$ w.r.t. the solution $x$, whose absolute value must be at most 1 since the decision maker can hold at most one item.

\begin{Claim}
    \label{clm:bounded-diff}
    Let $x$ be a feasible solution of \ref{program:relax}. For every $0 \leq \ell \leq r \leq T$, there must be
    \[
    \left|\sum_{i = \ell}^r q^b_i(x) - \sum_{i = \ell}^r q^s_i(x)\right| ~\leq ~ 1.
    \]
\end{Claim}

\begin{proof}
    We prove via contradiction. Assume for interval $[\ell, r]$, we have
    \[
    \sum_{i = \ell}^r q^b_i(x) - \sum_{i = \ell}^r q^s_i(x) ~>~ 1.
    \]
    If $r = T$, we may without loss of generality reduce $r$ to $T - 1$, as $q^b_T(x)$ must be $0$. When $r \leq T - 1$, note that
    \begin{align*}
        \sum_{i = 0}^r \sum_{j = r+1}^T x_{i, j} ~&\geq~  \sum_{i = \ell}^r \sum_{j = r+1}^T x_{i, j} \\
        ~&=~ \sum_{i = \ell}^r q^b_i(x) - \sum_{i = \ell}^{r - 1} \sum_{j = i + 1}^r x_{i,j} \\
        ~&\geq~ \sum_{i = \ell}^r q^b_i(x) - \sum_{i = \ell}^r q^s_i(x) ~>~ 1,
    \end{align*}
    which is in contrast to the assumption that $x$ is feasible as the constraint in the third line of \ref{program:relax} with $t = r + 1$ is not satisfied. Therefore, there must be $\sum_{i = \ell}^r q^b_i(x) - \sum_{i = \ell}^r q^s_i(x) ~\leq~ 1$.

    Symmetrically, when $\sum_{i = \ell}^r q^s_i(x) - \sum_{i = \ell}^r q^b_i(x) > 1$ holds, we may assume without loss of generality that $\ell \geq 1$, as $q^s_0(x)$ must be $0$. Then, note that
    \begin{align*}
        \sum_{i = 0}^{\ell - 1} \sum_{j = \ell}^T x_{i, j} ~&\geq~ \sum_{i = 0}^{\ell - 1} \sum_{j = \ell}^r x_{i, j} \\
        ~&=~ \sum_{j = \ell}^r q^s_j(x) - \sum_{i = \ell}^{r - 1} \sum_{j = i + 1}^r x_{i, j} \\
        ~&\geq~ \sum_{j = \ell}^r q^s_j(x) - \sum_{j = \ell}^r q^b_j(x) ~>~ 1,
    \end{align*}
    which violates the constraint in the third line of \ref{program:relax} with $t = \ell$.  Combining two cases proves \Cref{clm:bounded-diff}.
\end{proof}

Now, we are ready to prove \Cref{lma:init-decompose}.

\begin{proof}[Proof of \Cref{lma:init-decompose}]
We first show  \Cref{lma:init-decompose}.\ref{item:init-decompose-i}, i.e.,  within each interval $[\ell_d, r_d]$, the summation of $q^b_i(\tilde x)$ and $q^s_i(\tilde x)$ are both bounded by $2.5$. Note that if the desired index $r$ in Line 5 of \Cref{alg:init-decompose} does not exist, then at least one summation must be bounded by $1$, and \Cref{clm:bounded-diff} guarantees that another summation is bounded by 2, so \Cref{lma:init-decompose}.\ref{item:init-decompose-i} holds. We then consider the general case where index $r_d$ can be found. Since taking $r_d - 1$ is not feasible, there must be
\[
    \sum_{i = \ell_d}^{r_d - 1} q^b_i(\tilde x) ~<~ 1 \qquad \text{or} \qquad \sum_{i = \ell_d}^{r_d - 1} q^s_i(\tilde x) ~<~ 1.
\]
By \Cref{clm:bounded-diff}, we have
\begin{align*}
    \sum_{i = \ell_d}^{r_d - 1} q^s_i(\tilde x) + \sum_{i = \ell_d}^{r_d - 1} q^b_i(\tilde x) ~&=~ \min\left\{ \sum_{i = \ell_d}^{r_d - 1} q^s_i(\tilde x), \sum_{i = \ell_d}^{r_d - 1} q^b_i(\tilde x)\right\} + \max\left\{ \sum_{i = \ell_d}^{r_d - 1} q^s_i(\tilde x), \sum_{i = \ell_d}^{r_d - 1} q^b_i(\tilde x)\right\} \\
    ~&\leq~ 2 \cdot \min\left\{ \sum_{i = \ell_d}^{r_d - 1} q^s_i(\tilde x), \sum_{i = \ell_d}^{r_d - 1} q^b_i(\tilde x)\right\} + 1 ~<~ 3.
\end{align*}
Note that Constraint \eqref{eq:LP-marginals} guarantees 
\[
q^s_{r_d}(\tilde x) + q^b_{r_d}(\tilde x) ~\leq~ p_{r_d} ~\leq~ 1.
\]
Summing the above two inequalities together gives
\[
\sum_{i = \ell_d}^{r_d} q^s_i(\tilde x) + \sum_{i = \ell_d}^{r_d} q^b_i(\tilde x) ~\leq~ 4.
\]
Since \Cref{clm:bounded-diff} guarantees that
\[
\left |\sum_{i = \ell_d}^{r_d} q^s_i(\tilde x) - \sum_{i = \ell_d}^{r_d} q^b_i(\tilde x) \right| ~\leq~ 1,
\]
there must be
\[
    \sum_{i = \ell_d}^{r_d} q^b_i(\tilde x) ~\leq~ 2.5 \qquad \text{and} \qquad \sum_{i = \ell_d}^{r_d} q^s_i(\tilde x) ~\leq~ 2.5.
\]

Next, we prove  \Cref{lma:init-decompose}.\ref{item:init-decompose-ii}  via contradiction. Assume there exists $1 \leq i' < j' \leq T$ that satisfies $\tilde x_{i', j'} > 0$, such that $d' - d \geq 2$ when $i' \in [\ell_d, r_d]$ and $j' \in [\ell_{d'}, r_{d'}]$. Consider $j \in [\ell_{d+1}, r_{d+1}]$, and $i \leq j$ that satisfies $\tilde x_{i, j} > 0$. Recall that solution $\tilde x$ is given by \Cref{alg:uncrossing} after finishing the uncrossing process. Since we have $j \leq r_{d+1} < \ell_{d'} \leq j'$, there must be $i \leq i' \leq r_d < \ell_{d+1}$. Then, we have
\begin{align*}
    \sum_{i = 0}^{\ell_{d+1} - 1} \sum_{j = \ell_{d+1}}^T \tilde x_{i, j} ~&\geq~ \tilde x_{i', j'} + \sum_{i = 0}^{\ell_{d+1} - 1} \sum_{j = \ell_{d+1}}^{r_{d+1}} \tilde x_{i, j} \\
    ~&=~ \tilde x_{i', j'} + \sum_{j = \ell_{d+1}}^{r_{d+1}} q^s_j(\tilde x) \\
    ~&\geq~ \tilde x_{i', j'} + 1 > 1,
\end{align*}
where the first inequality uses the fact that $j' \notin [\ell_{d+1}, r_{d+1}]$, and therefore not counted in the following summation, the second equality uses the fact that $\tilde x_{i, j} > 0$ for $j \in [\ell_{d+1}, r_{d+1}]$ implies $i < \ell_{d+1}$, and the last line uses the fact that \Cref{alg:init-decompose} guarantees that $\sum_{i = \ell_{d+1}}^{r_{d+1}} q^s_i(\tilde x) \geq 1$ (note that $d+1 < d'$, so $r_{d+1}$ can't be $T$). Since the above inequality violates the constraint in the third line of \ref{program:relax} with $t = \ell_{d+1}$, the assumption that $(i', j')$ exists is in contrast to the feasibility of solution $\tilde x$. Therefore, \Cref{lma:init-decompose}.\ref{item:init-decompose-ii}  holds.
\end{proof}

Finally, we prove \Cref{lma:decompose-opt} by combining \Cref{clm:uncrossing} and \Cref{lma:init-decompose}.

\begin{proof}[Proof of \Cref{lma:decompose-opt}]
    To prove \Cref{lma:decompose-opt}, we decompose solution $\tilde x$ into three solutions $\tilde x^{(1)}$, $\tilde x^{(2)}$, and $\tilde x^{(3)}$, such that these three solutions satisfy the first and second statements in \Cref{lma:decompose-opt}. Then, we use \Cref{clm:uncrossing}, which states that $\lp(\tilde x) = \lp(x^*)$, to prove the third statement.

\vspace{0.5em} \noindent \textbf{Constructing $\tilde x^{(1)}$ and proof of \Cref{lma:decompose-opt}.\ref{item:decompose-opt-i}.} We first give the definition of $\tilde x^{(1)}$ and show that there exists an online algorithm achieving $\frac{1}{2} \cdot \lp(\tilde x^{(1)})$. For $0 \leq i < j \leq T$, we define
    \[
    \tilde x^{(1)}_{i, j} ~:=~ \tilde x_{i, j} \cdot \one[i = 0],
    \]
    i.e., $\tilde x^{(1)}$ represents all trades that buys (for free) from request $0$. Since solution $\tilde x^{(1)}$ only buys from request $0$ with at most one copy, the selling process given by $\tilde x^{(1)}$ can be viewed as the LP solution of a single-item Prophet Inequality problem, and a standard Prophet Inequality algorithm achieves an expected profit of $\frac{1}{2} \cdot \lp(\tilde x^{(1)})$. See e.g. \cite{correa2019recent} and its references for details.
    
\vspace{0.5em} \noindent \textbf{Constructing $\tilde x^{(2)}, \tilde x^{(3)}$ and proof of \Cref{lma:decompose-opt}.\ref{item:decompose-opt-ii}.} Given $\tilde x$, we run \Cref{alg:init-decompose} to get the initial decomposition $\calI$ for $\tilde x$, and we assume without loss of generality that $\calI = \{[\ell_1, r_1], \cdots [\ell_D, r_D]\}$, such that intervals are sorted in order, i.e., $\ell_1 = 1, r_D = T$, and $r_d + 1 = \ell_{d+1}$ for every $d \in [D - 1]$. We define $\tilde x^{(2)}$ and $\tilde x^{(3)}$ as follows:
\begin{align*}
    \tilde x^{(2)}_{i, j} ~:=~ \tilde x_{i, j} \cdot \one\big[i \in [\ell_d, r_d]: d \text{ is odd}\big], \quad \text{and} \quad \tilde x^{(3)}_{i, j} ~:=~ \tilde x_{i, j} \cdot \one\big[i \in [\ell_d, r_d]: d \text{ is even}\big].
\end{align*}

Now, we prove \Cref{lma:decompose-opt}.\ref{item:decompose-opt-ii}. We only prove the statement for $\tilde x^{(2)}$, as the proof for $\tilde x^{(3)}$ is identical. 
Define $\calI_{odd} = \{[\ell_d, r_d] \cup [\ell_{d+1}, r_{d+1}]: d \in [D]\land d \text{ is odd}\}$, where the last interval is refined as $[\ell_D, r_D]$ instead of $[\ell_D, r_D] \cup [\ell_{D+1}, r_{D+1}]$ when $D$ is odd. Since \Cref{lma:init-decompose} guarantees that for $(i, j): \tilde x_{i, j} > 0 \land i \in [\ell_d, r_d]$, there must be either $j \in [\ell_d, r_d]$ or $j \in [\ell_{d+1}, r_{d+1}]$ (if exists), for $(i, j): \tilde x^{(2)}_{i, j} > 0$, both $i$ and $j$ must belong to one interval in $\calI_{odd}$. Therefore, interval set $\calI_{odd}$ is a feasible decomposition of solution $\tilde x^{(2)}$.

It remains to show that for each interval $[\ell, r] \in \calI_{odd}$, solution $\left\{\tilde x^{(2)}_{i, j} \cdot \one\big[i, j \in [\ell, r]\big]\right\}$ is a $2.5$-trades solution. Assume $\ell$ equals to $\ell_d$, where $d$ is odd. By the definition of $\tilde x^{(2)}$ and \Cref{lma:init-decompose}.\ref{item:init-decompose-ii}, solution $\left\{\tilde x^{(2)}_{i, j} \cdot \one\big[i, j \in [\ell, r]\big]\right\}$ includes all trades of solution $\tilde x$ that buys from a request within interval $[\ell_d, r_d]$, i.e., we have
\begin{align*}
    \sum_{i, j \in [\ell, r]^2} \tilde x^{(2)}_{i, j} \cdot \one\big[i, j \in [\ell, r]\big]  ~=~ \sum_{i \in [\ell_d, r_d]} q^b_i(\tilde x) ~\leq~ 2.5,
\end{align*}
where the inequality follows from \Cref{lma:init-decompose}.\ref{item:init-decompose-i}. Therefore, $\left\{\tilde x^{(2)}_{i, j} \cdot \one\big[i, j \in [\ell, r]\big]\right\}$ is a $2.5$-trades solution, and solution $\tilde x^{(2)}$ admits a $2.5$-trades decomposition.

\vspace{0.5em} \noindent \textbf{Proof of \Cref{lma:decompose-opt}.\ref{item:decompose-opt-ii}.} 
    We note that it suffices to show $\tilde x^{(1)}+ \tilde x^{(2)}+ \tilde x^{(3)} = \tilde x$, as when the equality holds, we have 
    \[
    \lp(\tilde x^{(1)}) + \lp (\tilde x^{(2)}) + \lp(\tilde x^{(3)}) ~=~ \lp(\tilde x) ~=~ \lp(x^*), 
    \]
    where the last equality follows from \Cref{clm:uncrossing}. Note that for every $(i, j): \tilde x_{i, j} > 0$, we have
\begin{align*}
    \tilde x_{i, j} ~&=~ \tilde x_{i, j} \cdot \left(\one[i = 0] + \one[i \in [\ell_d, r_d]:d \text{ is odd}] + \one[i \in [\ell_d, r_d]:d \text{ is even}]\right) \\
    ~&=~ \tilde x^{(1)}_{i, j} + \tilde x^{(2)}_{i, j} + \tilde x^{(3)}_{i, j},
\end{align*}
where the first equality holds since the conditions in three indicator functions partition all possibilities for tuple $(i, j)$. Therefore, equality $\tilde x^{(1)}+ \tilde x^{(2)}+ \tilde x^{(3)} = \tilde x$ holds.
\end{proof}


%% file: iid.tex
\newcommand{\Expectsell}{\alpha_S}
\newcommand{\Expectbuy}{\alpha_B}

\section{Positive Results for Identical Request Distributions}
\label{sec:iid-upper}

In this section, we analyze the case with where the requests at each
time are i.i.d.\ draws from some distribution $\D$. This i.i.d.\
structure allows us to write a time-independent LP giving an upper
bound on the optimal value; we then show a generic rounding algorithm
whose performance is parameterized in terms of the expected number of
sales made by the algorithm (see \Cref{thm:iid-large-general}). As a
corollary of this result, we get the  positive results stated in
\Cref{thm:label-iid-small-informal}, \Cref{thm:label-iid-large-informal}, and \Cref{thm:unlabel-iid-large-informal}.

\subsection{A Stronger Benchmark and Linear-Programming Relaxation}

To better exploit the assumption that all request are i.i.d.\ draws
from the distribution $\D$, we compare our online algorithms in this
section to a stronger benchmark. Consider an offline algorithm with no
capacity constraint for intermediate steps, i.e., the decision maker
only needs to make sure that the final number of copies falls between
$[0, B]$, but it may hold more than $B$ or even fewer than $0$ copies
during intermediate steps; equivalently, this offline algorithm is
allowed to \emph{exchange the order} of all arriving requests to stay
within the capacity bounds of $[0,B]$. Let such an offline algorithm
be called an ``order-independent algorithm'', and let the
``order-independent profit'' be the optimal profit achieved by such an
order-independent algorithm. Since this is a relaxation, the optimal
order-independent profit is at least the profit of the optimal offline
algorithm with capacity constraints (which is not allowed to reorder
requests). Furthermore, since the request are i.i.d., the optimal
order-independent algorithm, performs the same operation on every
request in expectation. Given this observation, we can further relax
the optimal order-independent profit to a linear program: assume the
distribution $\D$ has support size $K$, such that the distribution
$\D$ generates price tuple $(b^{(k)}, s^{(k)})$ with probability
$p^{(k)}$, for $k \in [K]$. Our LP relaxation is:

\begin{alignat}{2}
  \text{maximize} \qquad 
  T \cdot \sum_{k \in [K]} \big(s^{(k)} \cdot y_k &- b^{(k)} \cdot z_k
  \big)  && \tag{$\text{LP}_{\textsc{IID}}$} \label{program:iid} \\ 
  \text{s.t.}\quad
  y_k + z_k 
  &\leq p^{(k)} &\qquad& \forall k \in [K],  \label{eq:lp-iid-marginals} \\
  B_0 + T \cdot \sum_{k \in [K]} (z_k - y_k) &\in [0, B] && \forall t \in [T],\label{eq:lp-iid-capacity} \\
  y_k, z_k &\geq 0. && \notag
\end{alignat}

Intuitively, the LP variables $y_k$ and $z_k$ represent the
probability of an order-independent algorithm selling or buying a
copy, when the type of a request is $k$. Constraints
\eqref{eq:lp-iid-marginals} ensure that the probability of buying or
selling when the type is $k$ cannot exceed $p^{(k)}$; constraints
\eqref{eq:lp-iid-capacity} ensure that the algorithm holds at most $B$
copies and at least $0$ copy when the process ends. 

Let $(y^*, z^*)$
be the optimal solution to \ref{program:iid}. Define
\begin{align}
    \Expectbuy ~:=~ \sum_{k \in [K]} z^*_k \qquad \text{and} \qquad
  \Expectsell ~:=~ \sum_{k \in [K]} y^*_k \label{eq:iid-bsprob-balance} 
\end{align}
to be the expected number of copies we buy and sell at each
step. Without loss of generality, we assume that for solution $(y^*,
z^*)$, constraint \eqref{eq:lp-iid-capacity} achieves its lower bound,
i.e.,
\begin{gather}
  B_0 + T \cdot \sum_{k \in [K]} (z^*_k - y^*_k) ~=~ B_0 + T \cdot
  (\Expectbuy - \Expectsell) ~=~0, \label{eq:1}
\end{gather}
as otherwise we can either find some $z^*_k > 0$ such that decreasing
the value of $z^*_k$ increases the objective function, or such $z^*_k$
does not exist and the instance becomes trivial where we always sell
during the whole process. This observation further implies that the optimal solution $(y^*, z^*)$ is not constrained by the upper bound $B$ in \eqref{eq:lp-iid-capacity}.

We further define
\begin{align}
    \Gamma ~:=~ T \cdot \Expectsell ~=~ B_0 + T \cdot \Expectbuy  \label{eq:iid-gamma-def}
\end{align}
to be the expected number of copies we sell during the whole process.

If $\lp(y, z)$ is the objective value of \ref{program:iid}
corresponding to a feasible solution $(y, z)$, the following lemma
shows that $\lp(y, z)$ is at least the expected profit of an optimal
offline algorithm to the \ATP problem in the i.i.d.\ case.

\begin{Lemma}[Relaxation]
  \label{lma:iid-lprelax}
  The linear program \ref{program:iid} is a relaxation of the \ATP
  problem for i.i.d.\ requests. In other words, let $(y^*, z^*)$ be
  the optimal solution to \ref{program:iid}, then $\lp(y^*, z^*)$ is
  at least the expected profit of the optimal offline algorithm for
  the \ATP problem.
\end{Lemma}

\begin{proof}
  Consider the structure of the optimal order-independent
  algorithm. Since the decision of such algorithm does not depend on
  the arrival order of requests, while every request is identically
  distributed, in expectation the optimal decision of the algorithm
  should only depend on the type of each request, rather than the
  order of their arrivals. Given this observation, we define without
  loss of generality $\tilde y_k$ to be the probability that the
  arriving request has type $k$ and the optimal order-independent
  algorithm buys at price $b^{(k)}$, and define $\tilde z_k$ to be the
  probability that the arriving request has type $k$ and the optimal
  order-independent algorithm sells at price $s^{(k)}$. Then, there
  must be $b^{(k)} + s^{(k)} \leq p^{(k)}$, as $p^{(k)}$ is the
  probability that type $k$ is realized. The inequality
  \[
    B_0 + T \cdot \sum_{k \in [K]} \tilde z_k - \tilde y_k ~\in [0, B]
  \]
  also holds, as the LHS represents the expected final inventory of
  the optimal order-independent algorithm. Then,
  $(\tilde y, \tilde z)$ is a feasible solution of \ref{program:iid},
  and therefore $\lp(y^*, z^*)$ upper-bounds the expected profit of
  the optimal offline algorithm, as we have
  $\lp(y^*, z^*) \geq \lp(\tilde y, \tilde z)$, while the optimal
  order-independent profit upper-bounds the expected profit of the
  optimal offline algorithm.
\end{proof}

Given \Cref{lma:iid-lprelax}, it suffices to provide an online
algorithm with an expected profit comparable to $\lp(y^*, z^*)$.

\subsection{A Unified Algorithm for  Large Capacity }

In this subsection, we prove the results for the setting where capacity
$B$ is sufficiently large. Our main technical theorem is
\Cref{thm:iid-large-general}; we later derive our results for both the
the symmetric case (where $b_t = s_t$ for all $t$) and the asymmetric
case from this theorem.

Consider \Cref{alg:iid-large} which performs a natural randomized
rounding, with one important change: as we approach the end of the
online process, we implement a sell-only strategy to clear the
remaining inventory to the extent possible. This sell-only stage is
necessary, as the benchmark $\lp(y^*, z^*)$ decomposes into two parts:
the revenue from selling the initial inventory $B_0$, and the trading
profit earned from $T$ requests. As the revenue from either part can
be arbitrarily larger than the other, we ensure that
\Cref{alg:iid-large} tries to clear its inventory when it ends.

\begin{algorithm}[tbh]
\caption{\textsc{Algorithm for IID Requests with Large Capacity}}
\label{alg:iid-large}
\begin{algorithmic}[1]
\State \textbf{Input:} Distribution $\D$ and optimal solution $(y^*, z^*)$ of \ref{program:iid}
\State Let $\tau = T - \left \lceil 13\log B \cdot \sqrt{T/\Expectsell} \right \rceil$.
\For{$t = 1, \ldots, T$}
\State Observe request $t$ with price tuple $(b_t, s_t) \sim \D$.
\State Let $k_t$ be the type of the request, i.e., say $(b_t, s_t) = (b^{(k_t)}, s^{(k_t)})$.
\State Draw $\rho_t \sim U[0, 1]$.
\If{$\rho_t \in [0, z^*_{k_t}/p^{(k_t)}]$ \textbf{and} $t \leq \tau$}
\State If the current inventory level is less than $B$, buy at $b^{(k_t)}$; otherwise, skip the request.
\EndIf
\If{$\rho_t \in [1 - y^*_{k_t}/p^{(k_t)}, 1]$}
\State If the current inventory level is not $0$, sell at $s^{(k_t)}$; otherwise, skip the request.
\EndIf 
\EndFor
\end{algorithmic}
\end{algorithm}

\begin{restatable}[General I.I.D.\ Theorem]{Theorem}{iidmain}
  \label{thm:iid-large-general}
  Given an \ATP instance with i.i.d.\ requests, \Cref{alg:iid-large}
  achieves a competitive ratio of
  $1 - O(\log B \cdot \Gamma^{-1/2} + B^{-1})$ against the optimal objective of \ref{program:iid}, provided $B_0 \geq 1$.
\end{restatable}

We defer the somewhat technical proof of \Cref{thm:iid-large-general}
to \Cref{sec:iid-mainalg-missing}; here is a high-level idea of the
proof.  Under the assumption that \Cref{alg:iid-large} clears all the
inventory, its performance depends on the number of \emph{blocked}
trades, i.e., buys that cannot happen when the inventory equals $B$,
and the sales that are blocked because the inventory reaches $0$. We
show that when $\Gamma$ is much smaller than $B^2$, since the standard
deviation of the whole trading process is $O(\sqrt{\Gamma})$, and
hence the number of blocked trades is bounded by $O(\sqrt{\Gamma})$;
on the other hand, when $\Gamma$ is much larger than $B^2$, the
trading process tends to stabilize, allowing us to show that the
inventory level at each timestep is close to being uniformly
distributed over $[0, B]$. In turn, this means that an $O(B^{-1})$
fraction of trades are blocked in expectation. Combining these two
cases proves \Cref{thm:iid-large-general}. While we defer the full
proof to the appendix, we now derive our results for the i.i.d.\ case
from it in the next subsections.

\subsection{Positive Results for  the Large Capacity Case}
\label{sec:posit-results-large-B}

We first use \Cref{thm:iid-large-general} to prove our positive
results in
\Cref{thm:label-iid-large-informal} and \Cref{thm:unlabel-iid-large-informal} for
the case of large budgets; the corresponding lower bounds of these
will be proved in \Cref{sec:lower}.

\begin{proof}[Proof of
  \Cref{thm:unlabel-iid-large-informal}\ref{item:asym-large-ub} and \ref{thm:label-iid-large-informal}\ref{item:sym-large-ub}]
    To prove the competitive ratio for the asymmetric case, we apply \Cref{alg:iid-large} with an extra pre-processing step: given an ATP instance with i.i.d. requests, we first solve \ref{program:iid} with the original capacity upper bound $B$, but run \Cref{alg:iid-large} with a  reduced capacity upper bound $B_0$. This pre-processing step is feasible, as we recall that the optimal solution $(y^*, z^*)$ of \ref{program:iid} is not constrained by the upper bound $B$ in constraint \eqref{eq:lp-iid-capacity}. Since $\Gamma = B_0 + T \cdot \Expectbuy \geq B_0$, \Cref{alg:iid-large} achieves at least $1 - O(\frac{\log B_0}{\sqrt{B_0}})$ times the $\lp(y^*, z^*)$, and therefore the competitive ratio of \Cref{alg:iid-large} is $1 - O(\frac{\log B}{\sqrt{B_0}})$.

  To prove the stronger bound of $1 - O(\nicefrac{(\log B)}{\sqrt{B}})$ for
  the symmetric case where $b^{(k)} = s^{(k)}$, our proof relies on
  the following observation: given an \ATP instance with symmetric
  prices, let $(y^*, z^*)$ be the optimal solution of
  \ref{program:iid}. If $y^*_k + z^*_k < p^{(k)}$ for some
  $k \in [K]$, increase $y^*_k$ and $z^*_k$ equally until constraint
  \eqref{eq:lp-iid-marginals} becomes tight; note that the constraint
  \eqref{eq:lp-iid-capacity} and the objective function value remain
  unchanged. Therefore, without loss of generality we may assume that
  constraint \eqref{eq:lp-iid-marginals} is tight, and therefore
  \[
    \Expectsell + \Expectbuy ~=~ \sum_{k \in [K]} (y^*_k + z^*_k) ~=~ \sum_{k \in [K]} p^{(k)} ~=~ 1,
  \]
  which implies
  \[
    2\Gamma ~=~ B_0 + T \cdot (\Expectsell + \Expectbuy) ~=~ B_0 + T ~\geq~ B;
  \]
  the final inequality uses the assumption from~\Cref{sec:our-model}
  that $B \leq B_0+T$. Therefore, the competitive ratio of
  \Cref{alg:iid-large} for the \ATP problem with symmetric prices
  becomes $1 - O(\frac{\log B}{\sqrt{B}})$.
\end{proof}

\subsection{Positive Results for the Unit Capacity Case}
\label{sec:posit-results-unit}

We now prove \Cref{thm:label-iid-small-informal}, our result for the
\ATP problem with i.i.d.\ requests. Given an \ATP instance with
request distribution $\D$, having unit capacity $B$ and unit initial
inventory level $B_0$, let $(y^*, z^*)$ be the optimal solution of
\ref{program:iid}. To begin, observe that~(\ref{eq:1})
implies that
\[ T \cdot \sum_k y^*_k = T\cdot \Expectsell \geq B_0 = 1. \]
Now define two solutions to (\ref{program:iid}):
\begin{alignat}{2}
  y'_k &= y^*_k/(T\Expectsell) & \qquad\text{and}\qquad & z'_k = 0, \\
  y''_k &= y^*_k - y'_k &\qquad\text{and}\qquad & z''_k = z^*_k. \label{eq:balanced-solution}
\end{alignat}



\begin{Lemma}
  \label{lma:unit-iid-decompose}
  The solutions $(y', z')$ and
  $(y'', z'')$ are both feasible solutions of \ref{program:iid}.
\end{Lemma}

\begin{proof}
  Since $T\Expectsell \geq 1$, we have that $y'_k \in [0, y^*_k]$ and
  hence $y''_k \leq y^*_k$. Therefore both the non-negativity
  constraints and constraint~(\ref{eq:lp-iid-marginals}) are satisfied
  for both solutions. To verify~(\ref{eq:lp-iid-capacity}) for
  $(y', z')$, note that
  $B_0 + T \sum_k (z'_k - y'_k) = B_0 - \sum_k y'_k = B - 1 = 0$.
  Finally,
  $B_0 + T \sum_k (z''_k - y''_k) = B_0 + T(\Expectbuy - \Expectsell)
  + 1 = 0 + 1 \leq B$, where we used the definition of $(y'',z'')$ for the
  first equality, and (\ref{eq:1}) for the second equality.
\end{proof}

\Cref{lma:unit-iid-decompose} shows the feasibility of $(y', z')$ and $(y'', z'')$. We now give rounding algorithms for the
two solutions $(y', z')$ and $(y'', z'')$. For solution $(y', z')$, we show in \Cref{lma:unit-iid-ocrs} that an independent rounding algorithm recovers a constant fraction of $\lp(y', z')$.

\begin{Lemma}
  \label{lma:unit-iid-ocrs}
  There exists an algorithm that rounds solution $(y', z')$ to get an
  expected profit of $\nicefrac12 \cdot \lp(y', z')$.
\end{Lemma}

\begin{proof}
  Observe that the fractional process corresponding to $(y', z')$
  performs no buying action, and the total number of (fractional)
  selling action is $1$. Therefore, solution $(y', z')$ can be viewed as the LP relaxation of a  single-item Prophet Inequality instance. To be specific, consider the following independent rounding algorithm: when price tuple $(b_t, s_t)$ is realized with type $k$, i.e., $(b_t, s_t) = (b^{(k)}, s^{(k)})$, we sell the unit inventory at price $s^{(k)}$ with probability $\nf{y'_k}{p^{(k)}}$, and terminate the process once this only copy is sold. Then, given that the inventory is not sold before $(b_t, s_t)$ is revealed, the probability that the inventory get sold is exactly
  \[
  \sum_{k \in [K]} p^{(k)} \cdot \frac{y'_k}{p^{(k)}} ~=~ \sum_{k \in [K]} p^{(k)} ~=~ \frac{\sum_{k\in [K]} y^*_k}{T \cdot \Expectsell} ~=~ T^{-1},
  \]
  and therefore the inventory remains unsold during the whole process is at least
  \[
  \left( 1 - \frac{1}{T} \right)^T ~\leq~ \frac{1}{e}.
  \]
  Therefore, the expected profit from this independent rounding process is at least
  \[
  \left(1 - \frac{1}{e}\right) \cdot \frac{\sum_{k \in [K]} p^{(k)} \cdot \frac{y'_k}{p^{(k)}} \cdot s^{(k)}}{T^{-1}} ~=~ \left(1 - \frac{1}{e}\right) \cdot \lp(y', z') ~\geq~ \frac{1}{2} \cdot \lp(y', z').
  \]
  where the equality follows from the fact that $z'_k = 0$ for every $k \in [K]$. 
\end{proof}

To round the solution $(y'',z'')$, we give the following \Cref{alg:unit-iid}.
\begin{algorithm}[tbh]
  \caption{\textsc{Algorithm for Solution $(y'', z'')$}}
  \label{alg:unit-iid}
  \begin{algorithmic}[1]
    \State \textbf{Input:} Solution $(y'', z'')$ defined in \eqref{eq:balanced-solution}.
    \State Drop the initial copy with probability $0.5$.
    \For{$t = 1, \ldots, T$}
    \State  Observe request $t$ with price tuple $(b_t, s_t) \sim \D$. Let $k_t$ be its type, i.e., $(b_t, s_t) = (b^{(k_t)}, s^{(k_t)})$.
    \State Draw $\rho_t \sim U[0, 1]$.
    \If{$\rho_t \in [0, z''_{k_t}/p^{(k_t)}]$}
    \State If the current inventory is $0$, buy at $b^{(k_t)}$; otherwise, skip the request.
    \EndIf
    \If{$\rho_t \in [1 - y''_{k_t}/p^{(k_t)}, 1]$}
    \State If the current inventory is $1$, sell at $s^{(k_t)}$; otherwise, skip the request.
    \EndIf 
    \EndFor
  \end{algorithmic}
\end{algorithm}

\begin{Lemma}
  \label{lma:unit-iid-main}
  \Cref{alg:unit-iid} gets an expected profit of
  $\nicefrac12 \cdot\lp(y'', z'')$.
\end{Lemma}

\begin{proof}
  Our proof relies on the following key observation: note that
  \begin{align*}
    \sum_{k \in [K]} y''_k ~=~ \sum_{k \in [K]} y^{*}_k - \sum_{k \in
    [K]} y'_k ~=~ \Expectsell - \frac{B_0}{T} ~=~ \frac{\Gamma -
    B_0}{T} ~=~ \Expectbuy ~=~ \sum_{k \in [K]} z''_k, 
  \end{align*}
  i.e., regarding each request, following solution $(y'', z'')$ leads
  to the same probability of buying or selling a copy. Then, the
  probability that the inventory is $0$ (or $1$) before every request
  $t \in [T]$ arrives must be $0.5$.

We prove the above statement via induction. The base case is $t = 1$. The statement holds as we drop the initial copy with probability $0.5$. Now assume $\pr[\text{inventory is }0 \text{ before }t] = 0.5$. For request $t$, we have
\begin{align*}
    \pr[\text{inventory is }0 \text{ before }t + 1] ~=~& \pr[\text{inventory is }0 \text{ before }t] \cdot \Big(1 - \sum_{k \in [K]} p^{(k)} \cdot \frac{z''_k}{p^{(k)}} \Big) \\
    ~&+ \pr[\text{inventory is }1 \text{ before }t] \cdot  \sum_{k \in [K]} p^{(k)} \cdot \frac{y''_k}{p^{(k)}} \\
    ~=~& \nicefrac12 \Big(1 - \sum_{k \in [K]} z''_k + \sum_{k \in [K]} y''_k \Big) ~=~ 0.5.
\end{align*}
Therefore, $\pr[\text{inventory is }0 \text{ before }t + 1] = 0.5$ also holds, which implies the probability that the inventory is $0$ (or $1$) before every request $t \in [T]$ arrives is $0.5$. Then, by linearity of expectation, the expected profit of \Cref{alg:unit-iid} is
\begin{align*}
    &\sum_{t \in [T]} \sum_{k \in [K]} p^{(k)} \cdot \frac{y''_k}{p^{(k)}} \cdot s^{(k)} \cdot \pr[\text{inventory is }1 \text{ before }t] \\
    &~~- \sum_{t \in [T]} \sum_{k \in [K]} p^{(k)} \cdot \frac{z''_k}{p^{(k)}} \cdot b^{(k)} \cdot \pr[\text{inventory is }0 \text{ before }t] \\
    ~=~& \nf{T}2 \cdot \Big(\sum_{k \in [K]} y''_k \cdot s^{(k)} -
         \sum_{k \in [K]} z''_k \cdot b^{(k)} \Big) ~=~ \nf12\cdot
         \lp(y''_k, z''_k).
         \qedhere
\end{align*}
\end{proof}

Finally, combining
\Cref{lma:unit-iid-decompose}, \Cref{lma:unit-iid-ocrs}, and \Cref{lma:unit-iid-main} together
proves \Cref{thm:label-iid-small-informal}.

\begin{proof}[Proof of \Cref{thm:label-iid-small-informal}]
  Running the algorithm from \Cref{lma:unit-iid-ocrs} with probability
  $\nicefrac12$, and \Cref{alg:unit-iid} otherwise, gives an expected
  profit of 
  \[
    \nicefrac12 \cdot \nicefrac12 \cdot \lp(y'_k, z'_k) + \nicefrac12
    \cdot \nicefrac12 \cdot \lp(y''_k, z''_k) ~=~ \nicefrac{1}{4}
    \cdot \lp(y^*, z^*),
  \]
  where the equality uses the fact that $y^*_k = y'_k + y''_k$ and
  $z^*_k = z'_k + z''_k$ for every $k \in [K]$.
\end{proof}


%% file: lower.tex
\section{Negative Results for the \ATP Problem}
\label{sec:lower}

In this section, we prove the hardness results for the \ATP
problem. We start with \Cref{thm:no-initial}, which states that no
competitiveness is achievable with no initial inventory, i.e., with
$B_0 = 0$.

\begin{proof}[Proof of \Cref{thm:no-initial}]
  Consider an \ATP instance with $B_0 = 0$ and $T = B$, where requests
  are i.i.d.\ draws from the following distribution $\D$:
  \begin{align*}
    (b, s) \sim D = 
    \begin{cases} 
      (1, 0) & \text{w.p. } 1 - \frac{0.1}{B} \\
      (+\infty, 2) & \frac{0.1}{B}  \\
    \end{cases}
  \end{align*}
  Note that if an online algorithm observes a price tuple $(1, 0)$ and
  decides to buy at price $1$, the probability that it can sell at
  price $2$ in the future is at most
  \[
    1 - \Big(1 - \frac{0.1}{B}\Big)^T ~\approx 1 - e^{-0.1} ~\leq~ 0.1.
  \]
  This suggests that an online algorithm should never buy at price
  $1$, since the expected profit is at most $-1 + 2 \cdot 0.1 < 0$;
  consequently, the optimal online algorithm has an expected profit of
  $0$. However, an optimal offline algorithm, which can see the
  future, can check if there is a price tuple $(+\infty, 2)$ following
  some occurrence of $(1,0)$---which happens with constant
  probability---and get unit profit in these scenarios. This implies
  that no online algorithm can have finite competitiveness with the
  offline optimum.
\end{proof}

\begin{proof}[Proof of \Cref{thm:label-iid-large-informal}\ref{item:asym-large-lb}]
  For the negative result in \Cref{thm:label-iid-large-informal}, note that when we have $b^{(k)} = +\infty$ for every price tuple $(b^{(k)}, s^{(k)})$ in the support of the identical request distribution $\D$, the \ATP problem degenerates to the $B_0$-unit Prophet Inequality problem with identical value distributions, leading to a $1 - \Omega(1/\sqrt{B_0})$ lower bound (see \cite{hajiaghayi2007automated}).
\end{proof}

\subsection{The Negative Result for Symmetric Prices}

Recall that we now need to prove
\Cref{thm:unlabel-iid-large-informal}\ref{item:sym-large-lb}, i.e., to
bound the performance of any online algorithm for the setting where
the prices are symmetric (i.e., $b_t = s_t)$) and the requests are
i.i.d.\ draws from some distribution $\D$.

Consider the following symmetric \ATP instance with $B = 1.5T = 3B_0$
for some large enough initial inventory level $B_0$; observe that we
can ignore the capacity bound $B$, because $B = B_0 + T$. (We assume
$T$ is suitably large for the lower bound.) The underlying
distribution $\D$ is defined as follows:
\begin{align*}
  (b, s) \sim D = 
  \begin{cases} 
    (1, 1) & \text{w.p. } 0.75 \\
    (0, 0) & \text{w.p. } 0.25 - 1/\sqrt{B}  \\
    (0.5, 0.5) & \text{w.p. } 1/\sqrt{B}
  \end{cases}
\end{align*}
The strategy that buys items at price $0$ and sells at $1$ (whenever
possible) makes $\Omega(B)$ profit, and this is trivially the best
possible. Therefore, to show that the competitive ratio of any online
algorithm is at most $1 - \Omega(1/\sqrt{B})$, it suffices to show
that any online algorithm incurs at least an $\Omega(\sqrt{B})$
(additive) regret against some offline algorithm---i.e., to show that
for any online algorithm, there exists an offline algorithm achieving
an expected profit of $\Omega(\sqrt{B})$ more than the online
algorithm.


For the proof, we define three events, which we show happen with
constant probability each.
For constants $c_2 \ll c_1 \ll 1$, the two
events are as follows:
\begin{itemize}
\item Event $\cE_0$: The number of occurrences of each price
  $0, \nf12, 1 $ in the time interval $[1,T/2]$ deviates from its
  expectation by at most $c_1\, \sqrt{B}$.
\item Event $\cE_1$: The number of occurrences of the ``low'' prices
  $0, \nf12$ during $[T/2+1, T]$ is low enough that following the
  strategy (at all times in $[1,T]$) of buying when a price $0$ or
  $\nf12$ arrives, and selling when the price $1$ arrives, ensures
  that the inventory never becomes negative at any moment, and that
  the final inventory level is in
  $[0.5c_2 \, \sqrt{B}, c_2 \, \sqrt{B}]$.
\item Event $\cE_2$: The number of occurrences of price $0$ during
  $[T/2+1, T]$ is high enough that following the strategy (at all
  times $T$) of buying when $0$ arrives, while selling when price $1$
  arrives, ensures that the inventory never becomes negative at any
  moment, and that the final inventory is again in
  $[0.5c_2 \, \sqrt{B}, c_2 \, \sqrt{B}]$.
\end{itemize}

\begin{Claim}
  \label{clm:const}
  The two events $\cF_1 := (\cE_0 \cap \cE_1)$ and
  $\cF_2 := (\cE_0 \cap \cE_2)$ happen with at least some absolute constant
  probability. Moreover, these events $\cF_1, \cF_2$ are disjoint.
\end{Claim}

\begin{proof}
  We give the proof for event $\cF_1$; the proof for $\cF_2$ is
  similar. The probability of $\cE_0$ is easy to analyze: standard
  Chernoff bounds ensure that the number of occurrences of each of the
  tuples deviates by more than $c_1 \sqrt{B}$ is at most some
  constant. So it remains to bound the probability of $\cE_1$ given
  $\cE_0$.
  Fix the realization of the first half. Let $B_t$ be the inventory
  after request $t$ is processed, assuming we follow the strategy of
  selling a copy whenever the price equals $1$, and buying a copy at
  all other ``low'' prices. Recall that $\cE_1$ requires two
  conditions:
  \begin{enumerate}
  \item[(a)] \textbf{Final Inventory Constraint:} The final
    inventory $B_T$ lands in $[0.5 c_2 \sqrt{B}, c_2 \sqrt{B}]$.
  \item[(b)] \textbf{Intermediate Constraint:} $B_t \ge 0$ for all
    $t \in (T/2, T]$.
  \end{enumerate}

  Conditioned on the starting inventory $B_{T/2}$, the Final Inventory
  Constraint requires the final inventory to be off by its expectation
  by $[c'\cdot \sqrt{B}, c'' \cdot \sqrt{B}]$ for some constants
  $c', c''$ satisfying $c'' - c' = 0.5\cdot c_2$. By a Chernoff bound,
  this happens with at least a constant probability.

  Next, we take the Intermediate Constraint into consideration. We
  show that the probability that the Final Inventory Constraint holds
  while the Intermediate Constraint breaks is small enough. Note that
  if we follow the strategy of selling a copy whenever the price is
  $1$ and buying a copy otherwise, in expectation the inventory
  decreases by $0.5$ after processing each request. Therefore, when
  the Final Inventory Constraint holds while the Intermediate
  Constraint breaks, it indicates that within at most $T$ steps after
  first hitting $B_t = 0$, this random walk process, which has a
  constant leftward drift in expectation, travels
  $0.5 c_2 \cdot \sqrt{B}$ steps to the right. Then, a standard random
  walk analysis can show that this event happens with probability at
  most $e^{-\Omega(\sqrt{B})}$. Therefore, when the realization of the
  first half is fixed, we have
  \begin{align*}
    &\pr[\text{Final Inventory OK}\land \text{Intermediate OK}] \\
    =~& \pr[\text{Final Inventory OK}] - \pr[\text{Final Inventory OK} \land \lnot \text{Intermediate OK}] \\
    \geq~&\text{constant}
  \end{align*}
  Finally, since the realizations of the two halves of the instance
  are independent, $\pr[\cE_1]$ is at least an absolute
  constant. Furthermore, an identical analysis can also show that
  $\pr[\cF_2]$ is at least an absolute constant.
\end{proof}

\begin{proof}[Proof of \Cref{thm:unlabel-iid-large-informal}]
  Given the symmetric $\ATP$ instance above.  Now, given any online
  algorithm, we construct the following offline benchmark.  When event
  $\cF_1$ happens, the offline algorithm simply buys whenever the
  price is $0$ or $0.5$, and sells whenever the price is~$1$.
  Similarly, when the disjoint event $\cF_2$ occurs, the offline
  algorithm buys whenever the price is~$0$ and sells whenever the
  price is~$1$. Otherwise, for all outcomes outside
  $\cF_1 \cup \cF_2$, we mimic the online algorithm. It suffices to
  show that the offline algorithm significantly outperforms the online
  algorithm on at least one of these events.
  
  We now analyze the performance loss of the online algorithm in these
  two cases.  Observe that at time~$T/2$, the online algorithm cannot
  distinguish, conditioned on satifying $\cE_0$, whether the
  realization satisfies $\cF_1$ or $\cF_2$---indeed, given the initial
  inventory level of $B_0$ and the fact that $T = 2B_0$, even selling
  at every single timestep would not cause inventory levels to be
  negative at time $T/2$. Moreover, given $\cE_0$, the number of timesteps
  with price $\nf12$ is at least
  \[ K := [(T/2) \cdot (1/\sqrt{B})] - c_1 \sqrt{B} = (1/3 - c_1)\cdot
    \sqrt{B}. \] We now consider whether the online algorithm
  purchases at least $K/2$ copies at price $\nf12$ by $T/2$.
  \begin{itemize}
  \item If it does buy at least $K/2$ copies at price $\nf12$, then its profit is less than the
    profit of the offline algorithm which only buys at 0 by at least $K/2$, for
    all outcomes in $\cE_2$.
  \item On the other hand, if the online algorithm purchases fewer
    than $K/2$ copies at price~$\nf12$ by time~$T/2$, then it must
    have rejected at least $K/2$ copies at price $\nf12$. Now the
    algorithm that buys all items at prices $0$ and $\nf12$ on
    outcomes in $\cE_1$ would buy all $K$  these copies, and still
    sell them, thereby making at least $K/2\cdot \nf12$ more profit.
  \end{itemize}
  Since $c_2, c_1 \ll 1$, this implies the regret of the online
  algorithm (versus an offline algorithm) is at least
  $\Omega(\sqrt{B}) \cdot \min(\pr[\cF_1], \pr[\cF_2])$, which is
  $\Omega(\sqrt{B})$ by \Cref{clm:const}, thereby completing the proof.
\end{proof}

%% file: ineq.tex
\section{Concentration Inequalities}

The following two concentration inequalities, for independent random
variables, and for martingales, will be used repeatedly in our analyses.

\begin{Theorem}[Bernstein's Inequality]
\label{Bernstein-bounded}
Let $X_1, \ldots, X_N$ be independent mean-zero random variables such
that $|X_i| \leq R$ for all $i$. Let $\sigma^2 := \sum_{i \in [N]} \E[X^2_i]$ be the variance of $\sum_{i \in [N]} X_i$. Then, for any $\epsilon \geq 0$,  we have
\begin{align*}
    \pr \left(\sum_{i = 1}^N X_i \geq \epsilon \right) ~&\leq~ \exp\left(- \frac{\epsilon^2/2}{\sigma^2 + R\epsilon/3}\right) \quad \text{and} \\
    \pr \left(\sum_{i = 1}^N X_i \leq -\epsilon \right) ~&\leq~ \exp\left(- \frac{\epsilon^2/2}{\sigma^2 + R\epsilon/3}\right) \enspace .
\end{align*}
\end{Theorem}

\begin{Theorem}[Freedman's Inequality]
\label{thm:freedman}
Let $\{Y_k\}_{k \ge 0}$ be a martingale with increments $\Delta Y_k = Y_k - Y_{k-1}$ satisfying $|\Delta Y_k| \le R$. Let $\langle Y \rangle_n = \sum_{k=1}^n \mathbb{E}[X_k^2 \mid \mathcal{F}_{k-1}]$ be the predictable quadratic variation. Then for any $\epsilon > 0$ and $v > 0$:
\begin{equation}
\pr\left[ \max_{1 \le k \le n} |Y_k| \ge \epsilon \text{ and } \langle
  Y \rangle_n \le v \right] ~\le~ 2 \exp\left( -\frac{\epsilon^2 /
    2}{v + R\epsilon / 3} \right) . \notag
\end{equation}
\end{Theorem}

%% file: single-noniid-missing.tex
\section{Omitted Proofs in \Cref{sec:unit-noniid}}
\label{sec:unit-noniid-missing}

In this section, we complete the missing proofs in \Cref{sec:unit-noniid}.

\subsection{Proof of Relaxation Lemma}
\begin{proof}[Proof of \Cref{lma:lprelax}]
It suffices to show that taking 
\[
x_{i, j} ~=~ \pr \left[\text{optimal offline solution buys from }i \text{ and sells to }j\right]
\]
gives a feasible solution of \ref{program:relax}.

Constraints \eqref{eq:LP-marginals} are feasible, as the probability that the optimal offline algorithm chooses an action of either buy or sell for request $i$ is at most $p_i$. Constraints \eqref{eq:lp-capacity} are feasible, as the quantity of $\sum_{i = 0}^{t-1} \sum_{j = t}^T x_{i, j}$ represents the expected number of the copy that the optimal offline algorithm carries after request $t$, which is at most $1$, as the offline algorithm can carry at most one copy at any time. Therefore, solution $x$ is feasible.
\end{proof}

\subsection{Proof of the $m$-Trades Lemma}

\begin{proof}[Proof of \Cref{lma:one-interval}]
    Note that it suffices to prove \Cref{lma:one-interval} only for $m = 1$, since when solution $x$ of \ref{program:relax} is an $m$-trades solution, 
    $x/m$ becomes a $1$-trade solution, and applying \Cref{lma:one-interval} to solution 
    $x/m$ gives the desired argument for $x$.

    When solution $x$ is a $1$-trade solution on interval $[\ell, r]$, there must be $\sum_{i = \ell}^r q^b_i(x) = \sum_{i = \ell}^r q^s_i(x) \leq 1$ (\Cref{def:mtrades}.\ref{item:mtrades-ii}), and the value of $q^b_i(x)$ and $q^s_i(x)$ can be non-zero only when $i \in [\ell, r]$ (\Cref{def:mtrades}.\ref{item:mtrades-i}). Our main idea is to apply a standard OCRS algorithm based on the values $q^b_i(x)$ and $q^s_i(x)$. To be specific, we show that \Cref{alg:one-trade} is the desired algorithm for \Cref{lma:one-interval}.

\begin{algorithm}[tbh]
\caption{\textsc{Trading Algorithm for $1$-trade solution}}
\label{alg:one-trade}
\begin{algorithmic}[1]
\State \textbf{Input:} $1$-trade solution $x$ of \ref{program:relax} on $[\ell, r]$.
\State Initial state: algorithm holds a copy with probability $0.5$.
\If{algorithm holds a copy}
\For{$i = \ell, \ldots, r$}
\If {price tuple $(b_i, s_i)$ arrives (with probability $p_i)$}
\State Sell the copy with probability 
            \[
            \frac{q^s_i(x)}{p_i} \cdot \frac{0.5}{1 - 0.5 \cdot \sum_{j = \ell}^{i-1} q^s_j(x)}. 
            \]
\State If the copy is sold, \textbf{break}.
\EndIf
\EndFor
\Else
\For{$i = \ell, \ldots, r$}
\If {price tuple $(b_i, s_i)$ arrives (with probability $p_i)$}
\State Buy a copy with probability 
            \[
            \frac{q^b_i(x)}{p_i} \cdot \frac{0.5}{1 - 0.5 \cdot \sum_{j = \ell}^{i-1} q^b_j(x)}. 
            \]
\State If algorithm buys a copy, \textbf{break}.
\EndIf
\EndFor
\EndIf
\end{algorithmic}
\end{algorithm}

\vspace{0.5em} \noindent \textbf{Proving the probability of holding a copy.} We first prove that after \Cref{alg:one-trade} is finished, the algorithm still keeps a copy in hand with probability $0.5$. To achieve this, we show that when the algorithm holds a copy when \Cref{alg:one-trade} begins, the copy is sold by \Cref{alg:one-trade} with probability $\frac{1}{2} \cdot \sum_{i = \ell}^r q^s_i(x)$. We prove by induction that for each $i\in [l, r]$, after request $i$ arrives, the copy remains available with probability $1 - \frac{1}{2} \cdot \sum_{j = \ell}^i q^s_j(x)$. The base case is $i = \ell - 1$, and the copy is available with probability $1$. Now assume the claim holds till request $i-1$. If the item remains available for request $i$, \Cref{alg:one-trade} sells the item with probability
    \[
    p_i \cdot \frac{q^s_i(x)}{p_i} \cdot \frac{0.5}{1 - 0.5 \cdot \sum_{j = \ell}^{i-1} q^s_j(x)} ~=~  \frac{0.5\cdot q^s_i(x)}{1 - 0.5 \cdot \sum_{j = \ell}^{i-1} q^s_j(x)}.
    \]
    Since whether the price tuple $(b_i, s_i)$ arrives and whether \Cref{alg:one-trade} sells the copy is independent to the decisions before request $i$, the copy remains available after request $i$ with probability
    \[
    \left(1 - \frac{1}{2} \cdot \sum_{j = \ell}^{i-1} q^s_j(x) \right) \cdot \left( 1 - \frac{0.5\cdot q^s_i(x)}{1 - 0.5 \cdot \sum_{j = \ell}^{i-1} q^s_j(x)} \right) ~=~ 1 - \frac{1}{2} \cdot \sum_{j = \ell}^{i} q^s_j(x),
    \]
    which proves the claim for the case after request $i$. Then, applying the claim for request $r$ shows that the item remains available with probability $1 - \frac{1}{2} \cdot \sum_{i = \ell}^r q^s_i(x)$, and equivalently the item is sold with probability $\frac{1}{2} \cdot \sum_{i = \ell}^r q^s_i(x)$.

    Symmetrically, when the algorithm does not hold a copy when \Cref{alg:one-trade} begins, it buys a copy with probability $\frac{1}{2} \cdot \sum_{i = \ell}^r q^b_i(x)$. Then, after \Cref{alg:one-trade} finishes, the algorithm holds a copy with probability
    \[
    0.5 \cdot \left(1 - \frac{1}{2} \cdot \sum_{i = \ell}^r q^s_i(x)\right) + 0.5 \cdot \frac{1}{2} \cdot \sum_{i = \ell}^r q^b_i(x) ~=~ 0.5,
    \]
    where we use the fact that $\sum_{i = \ell}^r q^s_i(x) = \sum_{i = \ell}^r q^b_i(x)$ in the last equality, which is true because every non-zero $x_{i,j}$ (buy at $i$, sell at $j$) satisfies $l\leq i<j\leq r$. So the total probability of buys and sells within the internal $[l,r]$ are equal.

\vspace{0.5em} \noindent \textbf{Proving the expected profit.} Next, we calculate the expected profit achieved by \Cref{alg:one-trade}. We first calculate the expected selling gain. Recall that an item remains available for request $i$ with probability  $1 - \frac{1}{2} \cdot \sum_{j = \ell}^{i-1} q^s_j(x)$, and the algorithm sells the item at price $s_i$ with probability $\frac{0.5\cdot q^s_i(x)}{1 - 0.5 \cdot \sum_{j = \ell}^{i-1} q^s_j(x)}$. Therefore, the total expected gain (assuming the algorithm holds a copy) is
\begin{align*}
    \sum_{i = \ell}^r  s_i \cdot \frac{0.5\cdot q^s_i(x)}{1 - 0.5 \cdot \sum_{j = \ell}^{i-1} q^s_j(x)} \cdot \left(1 - \frac{1}{2} \cdot \sum_{j = \ell}^{i-1} q^s_j(x)\right)  ~=~ \frac{1}{2} \cdot \sum_{i = \ell}^r s_i \cdot q^s_j(x).
\end{align*}
Symmetrically, the total expected gain assuming the algorithm does not hold a copy is
\begin{align*}
    \frac{1}{2} \cdot \sum_{i = \ell}^r b_i \cdot q^b_j(x).
\end{align*}
{Given that the decision maker holds a copy with probability $\frac{1}{2}$ before the algorithm begins,} the total expected profit achieved by \Cref{alg:one-trade} is
\begin{align*}
    \frac{1}{4} \cdot \sum_{i = \ell}^r (s_i \cdot q^s_j(x) - b_i \cdot q^b_j(x)) ~&=~ \frac{1}{4} \cdot \sum_{i = \ell}^r s_i \cdot \sum_{j = 0}^{i-1} x_{j,i} - \frac{1}{4} \cdot \sum_{i = \ell}^r b_i \cdot \sum_{j = i+1}^{r} x_{i,j} \\
    ~&=~ \frac{1}{4} \cdot \sum_{i = \ell}^r \sum_{j = i+1}^r x_{i, j} \cdot (s_j - b_i) ~=~ \frac{1}{4} \cdot \lp(x),
\end{align*}
{where the second-last equality is because $x_{i,j}=0$ if $i$ or $j$ are not in $[l, r]$.}
\end{proof}

Now, we apply \Cref{lma:one-interval} to prove \Cref{cor:multi-intervals}.

\begin{proof}[Proof of \Cref{cor:multi-intervals}]
    Let $x$ be a feasible solution of \ref{program:relax} that admits an $m$-trades decomposition corresponding to intervals $[\ell_1, r_1], \cdots, [\ell_D, r_D]$. We assume without loss of generality that the intervals are sorted in order. For $d \in [D]$, define solution $\hat x^{(d)}$ as $\hat x^{(d)}_{i, j} = x_{i, j} \cdot \one\big[i, j \in [\ell_d, r_d]\big]$. By definition, every $\hat x^{(d)}$ is an $m$-trades solution on $[\ell_d, r_d]$.
    
     Consider the algorithm that first buy a free copy from request $0$ with probability $0.5$, and then run \Cref{alg:one-trade} sequentially for solutions $\hat x^{(1)}, \cdots, \hat x^{(D)}$. Then, by \Cref{lma:one-interval}, the expected total profit of the above algorithm is
    \[
    \sum_{d \in [D]} \frac{1}{4m} \cdot \lp(\hat x^{(d)}) ~=~ \frac{1}{4m} \cdot \lp(x),
    \]
    where the equality follows from the definition of $m$-trades decomposition.
\end{proof}

\section{Generalizing Unit Capacity Case Beyond Bernoulli}
\label{sec:unit-noniid-generalize}


In this section, we briefly introduce how to extend the algorithms and proofs in \Cref{sec:unit-noniid} beyond Bernoulli distributions. Assume each distribution $\D_t$ has a bounded support size $K$ ($K$ types), such that for $k \in [K]$, the $t$-th price tuple becomes $(b^{(k)}_{t}, s^{(k)}_{t})$  with probability $p^{(k)}_{t}$. For the LP relaxation, we use $x_{i, j, k, k'}$ as the probability that the optimal offline algorithms buys from  $i$ with type $k$, and sells to  $j$ with type $k'$. Then, the LP relaxation can be correspondingly written as 
\begin{alignat}{2}
  \text{maximize} \qquad 
\sum_{i = 0}^{T-1} \sum_{j = i+1}^T \sum_{k \in [K]} \sum_{k' \in [K]}  x_{i, j, k, k'} \cdot &(s^{(k')}_{j} - b^{(k)}_{i}) &  \notag \\
  \text{s.t.}\quad
  \sum_{j = 0}^{i-1} \sum_{k' \in [K]} x_{j, i, k', k} + \sum_{j = i+1}^T \sum_{k' \in [K]} x_{i, j, k, k'}
    &\leq p^{(k)}_{i} && \forall i \in [T] \cup \{0\}, k \in [K],  \notag \\
    \sum_{i = 0}^{t-1} \sum_{j = t}^T \sum_{k \in [K]} \sum_{k' \in [K]} x_{i, j, k, k'} &\leq 1 && \forall t  \in [T], \notag \\
  x_{i,j,k,k'} &\geq 0. && \notag
\end{alignat}


Correspondingly, the definitions of $q^b_i(x)$ and $q^s_i(x)$ are generalized to $q^b_{i,k}(x)$ and $q^s_{i,k}(x)$, which are defined as
\[
q^b_{i,k}(x) ~:=~ \sum_{j = i+1}^T  \sum_{k' \in [K]} x_{i,j, k, k'} \qquad \text{and} \qquad q^s_{i,k}(x) ~:=~ \sum_{j = 0}^{i-1}  \sum_{k' \in [K]} x_{j, i, k, k'}
\]

The ideas of the remaining part of our algorithm are identical. For an $m$-trades solution on $[\ell, r]$, where the definition of $m$-trades solution corresponds to a solution on $[\ell, r]$ that satisfies
\[
\sum_{i = \ell}^{r-1} \sum_{j = i + 1}^r \sum_{k \in [K]} \sum_{k' \in [K]} x_{i, j, k, k'} ~\leq~ m,
\]
we use the generalized OCRS algorithm that works for beyond Bernoulli distributions to prove \Cref{lma:one-interval}. To be specific, if we initially hold a copy, when the price tuple $(b_i, s_i) = (b^{(k)}_i, s^{(k)}_i)$ arrives (with probability $p^{(k)}_i$), we sell the copy at price $s^{(k)}_i$ with probability
\[
\frac{q^s_{i,k}(x)}{p^{(k)}_i} \cdot \frac{0.5}{1 - 0.5 \cdot \sum_{j = \ell}^{i-1} \sum_{k' \in [K]} q^s_{j,k'}(x)}.
\]
Symmetrically, if we initially has no inventory, when the price tuple $(b_i, s_i) = (b^{(k)}_i, s^{(k)}_i)$ arrives (with probability $p^{(k)}_i$), we buy the copy at price $b^{(k)}_i$ with probability
\[
\frac{q^b_{i,k}(x)}{p^{(k)}_i} \cdot \frac{0.5}{1 - 0.5 \cdot \sum_{j = \ell}^{i-1} \sum_{k' \in [K]} q^b_{j,k'}(x)}.
\]
\Cref{lma:one-interval} can be proved analogously using the above OCRS algorithm. When \Cref{lma:one-interval} holds, the proof of \Cref{cor:multi-intervals} remains the same.

For the uncrossing step in \Cref{alg:uncrossing}, the condition in
Line 3 is generalized by also including the type of each
agent. Finally, the initial decomposition step in
\Cref{alg:init-decompose} and the proof of \Cref{lma:init-decompose}
are also identical, as the algorithm and related lemmas only rely on
the value of $q^b_i(x)$ and $q^s_i(x)$. Therefore, the algorithms and
proofs provided for Bernoulli distributions can be generalized to
distributions with support size $K$.


%% file: iid-mainalg-missing.tex
\section{Proof of \Cref{thm:iid-large-general}}
\label{sec:iid-mainalg-missing}

In this section, we prove \Cref{thm:iid-large-general}, which we now restate for convenience.

\iidmain*

The high level ingredients of the proof are the following. In
\Cref{lma:comp-ratio-bound}, we relate the competitive ratio of the
algorithm to two quantities: (a)~the final inventory level (i.e., the
number of unsold copies), and (b)~the number of times it is blocked
from buying an item because it is capacity-bound. We then bound each
of these in the following subsections, by relating these quantities to
a suitable random-walk process on the line.

\subsection{Notation}

To prove \Cref{thm:iid-large-general}, we first introduce some notation. Let
\[
\bar b ~:=~ \frac{\sum_{k \in [K]} b^{(k)} \cdot z^*_k}{\sum_{k \in [K]} z^*_k} ~=~ \frac{\sum_{k \in [K]} b^{(k)} \cdot z^*_k}{\Expectbuy} \quad \text{and} \quad \bar s ~:=~ \frac{\sum_{k \in [K]} s^{(k)} \cdot y^*_k}{\sum_{k \in [K]} y^*_k} ~=~ \frac{\sum_{k \in [K]} s^{(k)} \cdot y^*_k}{\Expectsell}
\]
be the expected buying/selling price when performing independent
rounding.
From \Cref{alg:iid-large}, recall the definition
\begin{gather}
  \tau := T - \lceil 13\log B \cdot
  \sqrt{T/\Expectsell} \rceil. \label{def:tau}
\end{gather}
For $t \in [T]$, let
\[
\xi_t ~:=~ -\one[1 - y^*_{k_t}/p^{(k_t)} \leq  \rho_t  \leq 1]  + \one[0 \leq \rho_t \leq z^*_{k_t}/p^{(k_t)}] \cdot \one[t \geq \tau + 1]
\]
represent the action we attempt to perform at time $t$, i.e., when
$\xi_t = -1$, it means we wish to sell at price $s^{(k_t)}$ if
inventory allows; when $\xi_t = +1$, it means we wish to buy at price
$b^{(k_t)}$ if inventory allows. It can be checked that when
$t \leq \tau$, we have
\[
\pr \left[ \xi_t = 1\right] ~=~ \sum_{k \in [K]} p^{(k)} \cdot \frac{z^*_k}{p^{(k)}} ~=~ \Expectbuy,
\]
and similarly we have $\pr[\xi_t = -1] = \Expectsell$ for every $t \in [T]$.

For $t \in [T] \cup \{0\}$, let $X_t$ be the inventory after request $t$ is processed, i.e., we have $X_0 = B_0$, and
\[
X_t = X_{t - 1} + \xi_t \cdot \one\big[X_{t-1} + \xi_t \in [0, B] \big] 
\]
for $t \in [T]$.

Finally, we define $N^+_t$ and $N^-_t$ to be the number of buys and sells we have performed till request $t$ respectively, and we define $M^+_t$ and $M^-_t$ to be the number of blocked buys and sells due to $[0, B]$ capacity constraint till request $t$, i.e., we have
\begin{align*}
    N^+_t ~=~ \sum_{t' = 1}^t \one[\xi_{t'} = 1 \land X_{t'-1} \neq B] \qquad &\text{and} \qquad N^-_t ~=~ \sum_{t' = 1}^t \one[\xi_{t'} = -1 \land X_{t'-1} \neq 0] \\
    M^+_t ~=~ \sum_{t' = 1}^t \one[\xi_{t'} = 1 \land X_{t'-1} = B] \qquad &\text{and} \qquad M^-_t ~=~ \sum_{t' = 1}^t \one[\xi_{t'} = -1 \land X_{t'-1} = 0].
\end{align*}

We remark that $\xi_t, X_t, U_t, N^+_t, N^-_t, M^+_t, M^-_t$ are all
random variables, and we will henceforth use their expectations to
characterize the performance of \Cref{alg:iid-large}.

\subsection{Proving \Cref{thm:iid-large-general}}

The following lemma reformulates the goal in \Cref{thm:iid-large-general}:
\begin{Lemma}[Reformulation Lemma]
\label{lma:comp-ratio-bound}
    The competitive ratio of \Cref{alg:iid-large} is lower-bounded by 
    \[
    1 - \frac{\E[X_T]}{B_0} - \frac{T - \tau}{T} - \frac{\E[M^+_\tau]}{\Gamma}.
    \]
\end{Lemma}
Since the value of the $(T-\tau)/T$ can be read off
from~(\ref{def:tau}), it remains to bound $\E[X_T]$ (the expected
final inventory level), and $\E[M^+_\tau]$ (the expected number of
blocked buys), which is done using the following two lemmas.

\begin{Lemma}[Final Inventory Bound]
  \label{lma:xt-bound}
    $\pr[X_T \leq \frac{T - \tau}{T} \cdot B_0] \geq 1 - 2B^{-2}$.
\end{Lemma}

\begin{Lemma}[Blocked Buys Bound]
    \label{lma:mbuy-bound}
    $\E[M^+_\tau] \leq \frac{\Gamma}{B} + 13\sqrt{\Gamma} \cdot \log B$.
\end{Lemma}

Before we prove these lemmas, let us use them to prove \Cref{thm:iid-large-general}.

\begin{proof}[Proof of \Cref{thm:iid-large-general}]    
    We first verify that \Cref{alg:iid-large} is feasible, which requires 
    \[
    \tau = T - \left \lceil 13\log B \cdot \sqrt{T/\Expectsell} \right \rceil \geq 0.
    \]
    Note that $\tau \geq 0$ must be satisfied when
    $T^2 \geq 14^2 \log^2 B \cdot \frac{T}{\Expectsell}$, which implies
    \[
    \Gamma ~\geq~ 196 \log^2 B.
    \]
    To fulfill this condition, we make an extra update before running \Cref{alg:iid-large}: if $B > 2\Gamma$, we set $B \gets 2\Gamma$.  Note that this update can only degrade the performance of \Cref{alg:iid-large}, while it does not change our proof of \Cref{thm:iid-large-general}, as the term $\log B/B$ is infinitesimal with respect to the term $\log B/\sqrt{\Gamma}$ when $B \geq \Omega(\Gamma)$. Therefore, assuming $B \leq 2\Gamma$ for \Cref{alg:iid-large} is without loss of generality,  and inequality $\Gamma ~\geq~ 196 \log^2 B$ holds when $\Gamma$ is sufficiently large, which is also without loss of generality as \Cref{alg:iid-large} trivially holds when $\Gamma$ is upper-bounded by a constant.
    
    Next, we prove the competitive ratio of \Cref{alg:iid-large}. By \Cref{lma:xt-bound} and the fact that $X_T \in [0, B]$, the value of $\E[X_T]$ is bounded by 
    \[
    \E[X_T] ~\leq~ (1 - 2B^{-2}) \cdot \frac{T - \tau}{T} \cdot B_0 + 2B^{-2} \cdot B ~\leq~ \frac{T - \tau}{T} \cdot B_0 + \frac{2B_0}{B},
    \]
    where the last inequality uses the assumption that $B_0 \geq 1$. Then, by \Cref{lma:comp-ratio-bound}, the competitive ratio of \Cref{alg:iid-large} is at least 
    \[
    1 - \frac{\E[X_T]}{B_0} - \frac{T - \tau}{T} - \frac{\E[M^+_\tau]}{\Gamma} ~\geq~ 1 - \frac{2(T - \tau)}{T} - \frac{3}{B} - \frac{13\log B}{\sqrt{\Gamma}} ~\geq~ 1 - \frac{3}{B} - \frac{41 \log B}{\sqrt{\Gamma}},
    \]
    where the last inequality uses the fact that 
    \[
    \frac{T - \tau}{T} ~=~ \frac{\left \lceil 13\log B \cdot \sqrt{T/\Expectsell} \right \rceil}{T} ~\leq~ \frac{14\log B \cdot \sqrt{T/\Expectsell}}{T} ~=~ \frac{14 \log B}{\sqrt{\Gamma}}. \qedhere
    \]
\end{proof}

We now prove \Cref{lma:comp-ratio-bound,lma:xt-bound,lma:mbuy-bound} in
the following subsections.

\subsection{Proof of the Reformulation Lemma (\Cref{lma:comp-ratio-bound})}

\begin{proof}[Proof of \Cref{lma:comp-ratio-bound}]
    In \Cref{alg:iid-large}, we first realize each request's type, and then decide our action via independent rounding. Since these two sources of randomness are independent, it's equivalent to first realize our decision $\xi_t$, and then we realize the type of current request to decide the buying/selling price. By conditional expectation, when $\xi_t = -1$, the expected selling price is exactly
    \[
    \frac{\sum_{k \in [K]}  p^{(k)} \cdot \frac{y^*_k}{p^{(k)}} \cdot s^{(k)} }{\sum_{k \in [K]} p^{(k)} \cdot \frac{y^*_k}{p^{(k)}}} ~=~ \bar s,
    \]
    and similarly, when $\xi_t = 1$, the expected buying price is exactly $\bar b$. 

    Note that when $\xi_t$ is fixed, whether the buy/sell action is proceeded only depends on the current inventory, instead of the buying/selling price. Since $N^+_T$ and $N^-_T$ represents the number of buys and sells we perform, the expected profit of \Cref{alg:iid-large} can be written as
    \begin{align*}
         \E[N^-_T] \cdot \bar s - \E[N^+_T] \cdot \bar b. 
    \end{align*}
    Note that for any realization of the instance, there must be $X_0 + N^+_t - N^-_t = X_T$. By replacing $N^-_t$ by $X_0 + N^+_t - X_T$, we further simplify the above expression as
    \begin{align}
        \big(B_0 + \E[N^+_T] - \E[X_T] \big) \cdot \bar s - \E[N^+_T] \cdot \bar b ~=~ \big(B_0 - \E[X_T] \big) \cdot \bar s + \E[N^+_\tau] \cdot (\bar s - \bar b), \label{eq:alg-bound}
    \end{align}
    where the equality uses the observation that \Cref{alg:iid-large} no longer buys a copy after request $\tau$, and therefore $N^+_T = N^+_\tau$. We further note that 
    \[
    N^+_\tau ~=~ \sum_{t \in [\tau]} \one[\xi_t = 1] - M^+_\tau.
    \]
    Since $\pr[\xi_t = 1] = \Expectbuy$ for every $t \leq \tau$, taking expectation over the above inequality and applying to \eqref{eq:alg-bound} further simplifies the expected profit  of \Cref{alg:iid-large} as
    \begin{align}
         \big(B_0 - \E[X_T] \big) \cdot \bar s + \big(\Expectbuy \cdot \tau - \E[M^+_\tau]\big) \cdot (\bar s - \bar b). \label{eq:alg-bound-final}
    \end{align}

    Next, we rewrite the value of $\lp(y^*, z^*)$ with  newly defined notations. We have
    \begin{align}
        \lp(y^*, z^*) ~&=~ T \cdot \sum_{k \in [K]} \big(s^{(k)} \cdot y_k - b^{(k)} \cdot z_k \big) \notag \\
        ~&=~ T \cdot \Expectsell \cdot \bar s - T \cdot \Expectbuy \cdot \bar b \notag \\
        ~&=~ B_0 \cdot \bar s + T \cdot \Expectbuy  \cdot (\bar s - \bar b), \label{eq:lp-bound-final}
    \end{align}
    where the second line follows the definition of $\Expectbuy, \Expectsell, \bar b, \bar s$, and the third line uses the observation that $T \cdot \Expectsell = \Gamma =  T \cdot \Expectbuy + B_0$.  Subtracting \eqref{eq:alg-bound-final} from \eqref{eq:lp-bound-final} gives the profit gap between \Cref{alg:iid-large} and $\lp(x^*, y^*)$ as
    \begin{align}
        \E[X_T] \cdot \bar s + \big( (T - \tau) \cdot \Expectbuy + \E[M^+_\tau] \big) \cdot (\bar s - \bar b) \label{eq:loss-bound-final}
    \end{align}
    Therefore, the competitive ratio of \Cref{alg:iid-large} is at least
    \begin{align*}
        1 - \frac{\text{expr.}~\eqref{eq:loss-bound-final}}{\text{expr.}~\eqref{eq:lp-bound-final}} ~&\geq~ 1 - \frac{\E[X_T] \cdot \bar s}{B_0 \cdot \bar s} - \frac{(T - \tau) \cdot \Expectbuy \cdot (\bar s - \bar b)}{T \cdot \Expectbuy \cdot (\bar s - \bar b)} - \frac{\E[M^+_\tau]\cdot (\bar s - \bar b)}{(B_0 + T \cdot \Expectbuy) \cdot (\bar s - \bar b)} \\
        ~&\geq~1 - \frac{\E[X_T]}{B_0} - \frac{T - \tau}{T} - \frac{\E[M^+_\tau]}{\Gamma}. \qedhere
    \end{align*}
\end{proof}

\subsection{Bounding the Final Inventory Level (\Cref{lma:xt-bound})}

Our proof of \Cref{lma:xt-bound} relies on the following two lemmas;
the first bounding $X_\tau$ with high probability, and the latter
showing that conditioned on the bound on $X_\tau$, the value of $X_T$
is also small (i.e., the final inventory is almost cleared). 
\begin{Lemma}
\label{lma:xtau-bound}
    $\pr[X_\tau  \leq 12\sqrt{\Gamma} \cdot \log B + \frac{T -
      \tau}{T} \cdot B_0] \geq 1 - B^{-2}$.
\end{Lemma}

\begin{Lemma}
\label{lma:clear-inv-bound}
    $\pr[X_T \leq \frac{T - \tau}{T} \cdot B_0 \mid X_\tau  \leq
    12\sqrt{\Gamma} \cdot \log B + \frac{T - \tau}{T} \cdot B_0] \geq 1 - B^{-2}$.
\end{Lemma}



Since the statement of~\ref{lma:xt-bound} follows from
\Cref{lma:xtau-bound,lma:clear-inv-bound} by a na\"{\i}ve union bound,
it now suffices to prove these two lemmas. We begin by bounding the
inventory level at time $\tau$.

\paragraph{Proving \Cref{lma:xtau-bound}.} Consider the following two random walk processes: let $\{U_t\}_{t \in [\tau]}$ be the unbounded random walk process w.r.t $\xi_t$, and $\{W_t\}_{t \in [\tau]}$ be the left-bounded random walk process w.r.t.\ $\xi_t$. I.e., we have $U_0 = W_0 = B_0$, and for $t \in [\tau]$, we have
\[
U_t ~=~ U_{t-1} + \xi_t, \qquad \text{and} \qquad W_t = W_{t - 1} + \xi_t \cdot \one\big[W_{t-1} + \xi_t \geq 0\big].
\]
 Then, the following \Cref{clm:xt-to-ut} leads us to instead bound the randomness of $U_t$:

\begin{Claim}
\label{clm:xt-to-ut}
    For any realization of $\{\xi_t\}_{t \in [\tau]}$, we have
    \[
    X_\tau ~\leq~ U_\tau - \min \Big\{0,~ \min_{t \in [\tau]} U_t \Big\}.
    \]
\end{Claim}

\begin{proof}
    We first observe that there must be $W_t \geq X_t$ for every $t \in [\tau]$. We show this observation via contradiction: suppose $t \in [\tau]$ is the smallest index that satisfies $W_t < X_t$. Since all the random walk processes we defined move by at most one step at each time increment, there must be $W_{t - 1} = X_{t-1}$, which implies
    \[
    X_t - X_{t - 1} ~=~ \xi_t \cdot \one\big[X_{t-1} + \xi_t \in [0, B] \big] ~\leq~ \xi_t \cdot \one\big[W_{t-1} + \xi_t \geq 0 \big] ~=~ W_t - W_{t-1},
    \]
    which is in contrast to the assumption of $W_t < X_t$. Therefore, there must be $W_\tau \geq X_\tau$.

    We further show that 
    \begin{align}
        W_\tau ~\leq~ U_\tau - \min \Big\{0,~ \min_{t \in [\tau]} U_t \Big\} \label{eq:wt-bound}
    \end{align}
    to finish the proof of \Cref{clm:xt-to-ut}. Fix the realization of $\{\xi_t\}_{t \in [\tau]}$. Let $U_m = \min_{t \in [\tau]} U_t$ be the minimum value attained by $U_t$. If $U_m \geq 0$, then there must be $W_t = U_t$ for every $t \in [\tau]$, and \eqref{eq:wt-bound} holds. Otherwise, we define   $\widetilde U_t$ and $\widetilde W_t$ to be the unbounded and left-bounded random walk processes with a shifted initial position, i.e., we have   $\widetilde U_0 = \widetilde W_0 = B_0 - U_m$, while 
    \[
    \widetilde U_t ~=~ \widetilde U_{t-1} + \xi_t, \qquad \text{and} \qquad \widetilde W_t = \widetilde W_{t - 1} + \xi_t \cdot \one\big[\widetilde W_{t-1} + \xi_t \geq 0\big]
    \]
    for $t \in [\tau]$. Note that condition $\widetilde W_{t-1} + \xi_t \geq 0$ always holds. If not, let $t$ be the minimum index such that $\widetilde W_{t - 1} + \xi_t < 0$. Then, there must be $\widetilde W_{t - 1} = \widetilde U_{t - 1} = 0$ while $\xi_t = -1$, and therefore $\widetilde U_t = -1$. This implies
    \begin{align*}
        U_t ~=~ \widetilde U_t - \widetilde U_0 + U_0  ~=~ -1 - (B_0 - U_m) + B_0 ~=~ U_m - 1,
    \end{align*}
    which is in contrast to the assumption that $U_m = \min_{t \in [\tau]} U_t$. Therefore, $\one\big[\widetilde W_{t-1} + \xi_t \geq 0\big] = 1$ always holds, implying that $\widetilde W_t = \widetilde U_t$.

    Finally, note that $\widetilde W_t \geq W_t$ always holds, as $\widetilde W_t$ and $W_t$ share the same update process, while $\widetilde W_0 = W_0 - U_m > W_0$ (recall that we assume $U_m < 0$). Then, 
    \begin{align*}
        W_\tau ~\leq~ \widetilde W_\tau ~=~ \widetilde U_\tau ~=~ U_\tau - U_0 + \widetilde U_0 ~=~ U_\tau - U_m,
    \end{align*}
    which proves \eqref{eq:wt-bound}.
\end{proof}

Now, we are ready to prove \Cref{lma:xtau-bound}.

\begin{proof}[Proof of \Cref{lma:xtau-bound}]
    Recall that for $t \in [\tau]$, $\xi_t$ equals to $-1$ with probability $\Expectsell$, $1$ with probability $\Expectbuy$, and $0$ otherwise. Then, for $t \in [\tau]$, we have
    \[
    \E[U_t] = B_0 + \sum_{t' \in [t]} \E[\xi_{t'}] ~=~ B_0 - t \cdot (\Expectsell - \Expectbuy).
    \]
    
    Define $Y_t = U_t - \E[U_t]$. Then,  $\{Y_t\}_{t \in [\tau]}$ is a martingale with increments
    \[
    \Delta Y_t ~=~ Y_t - Y_{t-1} ~=~ \xi_t - (\Expectsell - \Expectbuy).
    \]
    Then, the absolute value of the increments is bounded by
    \[
    |\Delta Y_t| ~\leq~ |\xi_t| + |\Expectsell - \Expectbuy| ~\leq~ 2.
    \]
    Furthermore, since $\xi_t$ is a random variable independent to the history till $t - 1$, the predictable quadratic variation $\langle Y \rangle_t$ satisfies
    \begin{align*}
        \langle Y \rangle_t ~=~ \sum_{t' \in [t]} \E \left[ \big(\xi_{t'} - (\Expectsell - \Expectbuy)\big)^2\right] ~=~ t \cdot  \E[\xi^2_t] - t \cdot (\Expectsell - \Expectbuy)^2 ~\leq~ t \cdot (\Expectsell + \Expectbuy) ~\leq~ 2 \Gamma,
    \end{align*}
    where the second equality uses the fact that every $\xi_t$ is i.i.d. with expectation $\Expectsell - \Expectbuy$, and the last inequality uses the fact that $t \leq \tau \leq T$, while $T \cdot (\Expectsell + \Expectbuy) \leq 2T \cdot \Expectsell = 2\Gamma$. Let $\epsilon = 6\sqrt{\Gamma} \cdot \log B$. By Freedman's Inequality (\Cref{thm:freedman}), we have
\begin{align}
    \pr\left[ \max_{t \in [\tau]} |Y_t| \ge \epsilon \text{ and } \langle Y \rangle_\tau \le 2\Gamma \right] \le 2 \exp\left( -\frac{\epsilon^2 / 2}{2\Gamma + 2\epsilon / 3} \right) \leq 2\exp(-3\log B) \leq B^{-2}, \label{eq:xtau-freedman}
\end{align}
where the second inequality uses the fact that
\[
\epsilon^2/2 ~=~ 18\Gamma \cdot \log^2 B ~\geq~ 3\log B \cdot (2\Gamma + 4\sqrt{\Gamma} \cdot \log B),
\]
and the last inequality holds when $B$ is sufficiently large. Then, we have
\begin{align*}
    &\pr\left[ X_\tau  \leq 12\sqrt{\Gamma} \cdot \log B + \frac{T - \tau}{T} \cdot B_0\right] \\
    \geq~& \pr \left[U_\tau \leq 6 \sqrt{\Gamma} \cdot \log B + \frac{T - \tau}{T} \cdot B_0 ~\land~  \min_{t \in [\tau]} U_t \geq -6 \sqrt{\Gamma} \cdot \log B \right] \\
    \geq~& \pr \left[Y_\tau \leq 6 \sqrt{\Gamma} \cdot \log B  ~\land~  \min_{t \in [\tau]} Y_t \geq -6 \sqrt{\Gamma} \cdot \log B \right] \\
    \geq~&  \pr\left[ \max_{t \in [\tau]} |Y_t| \leq 6\sqrt{\Gamma} \cdot \log B  \right] \\
    \geq~& 1 - \pr\left[ \max_{t \in [\tau]} |Y_t| \geq 6\sqrt{\Gamma} \cdot \log B =\epsilon  ~\land~ \langle Y \rangle_\tau \le 2\Gamma \right] ~\geq~ 1 - B^{-2},
\end{align*}
where the first inequality uses \Cref{clm:xt-to-ut}; the second inequality uses the fact that for every $t \in [\tau]$, we have
\[
\E[U_t] ~\geq~ \E[U_\tau] ~=~ B_0 -\tau \cdot (\Expectsell - \Expectbuy) ~=~ B_0 - \tau \cdot \frac{B_0}{T} ~=~ \frac{T - \tau}{T} \cdot B_0 ~\geq~ 0,
\]
and therefore $Y_\tau \leq \epsilon$ implies $U_\tau \leq \epsilon + \E[U_\tau]$, while $Y_t \geq -\epsilon$ implies $U_t \geq -\epsilon$; the last line uses \eqref{eq:xtau-freedman} and the fact that $\langle Y \rangle_\tau \le 2\Gamma$ always holds.
\end{proof}

\paragraph{Proving \Cref{lma:clear-inv-bound}.} Finally, we prove \Cref{lma:clear-inv-bound} to complete the proof of \Cref{lma:xt-bound}.

\begin{proof}[Proof of \Cref{lma:clear-inv-bound}]

    When $t \in [\tau + 1, T]$, \Cref{alg:iid-large} attempts to sell one copy with probability $\Expectsell$. To prove \Cref{lma:clear-inv-bound}, it suffices to show that the total number of selling attempts when $t \in [\tau + 1, T]$ is at least $12 \sqrt{\Gamma} \cdot \log B$, i.e., we aim at showing
    \begin{align}
        \pr \left[\sum_{t = \tau + 1}^T \one[\xi_t = -1] ~\geq~ 12 \sqrt{\Gamma} \cdot \log B \right] ~\geq~ 1- B^{-2}. \label{eq:clear-inv-bound-target}
    \end{align}

    We prove \eqref{eq:clear-inv-bound-target} via Bernstein's Inequality. Note that each $\one[\xi_t = -1]$ is an independent variable that equals to $1$ with probability $\Expectsell$, and $0$ otherwise. Then, we have $\E[\one[\xi_t = -1]] = \Expectsell$, and the variance of the summation of mean-zero random variables $\one[\xi_t = -1]- \Expectsell$ can be bounded as
    \begin{align*}
        \var \left(\sum_{t = \tau + 1}^T \one[\xi_t = -1] - \Expectsell \right) ~&=~ \var \left(\sum_{t = \tau + 1}^T \one[\xi_t = -1] \right) \\
        ~&\leq~ \sum_{t = \tau + 1}^T \E\big[\one^2[\xi_t = -1]  \big]\\
        ~&=~ (T - \tau) \cdot \Expectsell ~\leq~ 14\sqrt{\Gamma} \cdot \log B.
    \end{align*}
    Then, we have
    \begin{align*}
        \pr \left[\sum_{t = \tau + 1}^T \one[\xi_t = -1] < 12 \sqrt{\Gamma} \cdot \log B \right] ~&\leq~ \pr \left[\sum_{t = \tau + 1}^T \big(\one[\xi_t = -1] - \Expectsell \big) ~<~ -\sqrt{\Gamma} \cdot \log B \right]  \\
        ~&\leq~ \exp \left(-\frac{\Gamma \cdot \log^2 B/2}{14\sqrt{\Gamma} \cdot \log B + \sqrt{\Gamma} \cdot \log B / 3}\right) \\
        ~&\leq~ \exp \left(-\frac{ \sqrt{\Gamma }\cdot \log B}{30}\right)  ~\leq~ B^{-2},
    \end{align*}
    where the first inequality uses the fact that $(T - \tau) \cdot \Expectsell ~\geq~ 13 \sqrt{\Gamma} \cdot \log B$;
    the second inequality follows Bernstein's Inequality (\Cref{Bernstein-bounded}), which is applicable as every random variable $\one[\xi_t = -1] - \Expectsell$ is mean-zero;  the last line holds when $\Gamma$ is sufficiently large. Taking the complement of the above inequality gives \eqref{eq:clear-inv-bound-target}.
\end{proof}

\subsection{Bounding the Number of Blocked Buys (\Cref{lma:mbuy-bound})}

Before proving \Cref{lma:mbuy-bound}, we first introduce the following \Cref{lma:balanced-walk}, which converts the value of $\E[M^+_\tau]$ into a balanced random-walk process and can further simplify our analysis.

\begin{Lemma}
    \label{lma:balanced-walk}
    Let $\{\widetilde X_t\}_{t \in [T]}$ be a $T$-steps random walk process with laze boundaries $[0, B]$. To be specific, let $\widetilde X_0 = 0$. For $t \in [T]$, let $\tilde \xi_t$ be an independent random variable such that $\pr[\tilde \xi_t = 1] = \pr[\tilde \xi_t = -1] = \bar \alpha$, and $\xi_t = 0$ otherwise, where $\bar \alpha = 0.5 \cdot (\Expectbuy + \Expectsell)$. For $t \in [T]$, the update rule of the random walk process follows
    \[
    \widetilde X_t = \widetilde X_{t - 1} + \tilde \xi_t \cdot \one\big[\widetilde X_{t - 1} + \tilde \xi_t \in [0, B] \big].
    \]
    Let $\widetilde X^-_T$ be the total number of rounds blocked by the left boundary, i.e., the number of rounds such that $\widetilde X_{t - 1} = 0$ while $\tilde \xi_t = -1$. Then, we have $\E[M^+_\tau] \leq \E[\widetilde X^-_T]$.
\end{Lemma}

\begin{proof}
    We prove \Cref{lma:balanced-walk} via constructing a series of
    random-walk processes. Each of these random walks are restricted
    to $[0, B]$---any steps that cause them to move outside this
    interval are ``blocked'' and cause the position to remain
    unchanged. 
    We first initiate $\{X^{(1)}_t\}_{t \in [\tau]}$ to be the process equivalent to $\{X_t\}_{t \in [\tau]}$, i.e., we have $X^{(1)}_0 = B_0$, $\xi^{(1)}_t$ be an independent random variable that equals to $-1$ with probability $\Expectsell$, equals to $1$ with probability $\Expectbuy$, and equals to $0$ otherwise. The update rule follows $X^{(1)}_t = X^{(1)}_{t - 1} + \xi^{(1)}_t \cdot \one[X^{(1)}_{t - 1} + \xi^{(1)}_t \in [0, B]]$. Let $M^{(1)}_T$ be the number of rounds the random walk step is blocked by the upper boundary $B$, i.e., the number of $t \in [T]$ that satisfies $X^{(1)}_{t - 1} = B$ while $\xi^{(1)}_t = +1$. Then, there must be $\E[M^+_\tau] = \E[M^{(1)}_\tau]$.

    Now, we start to modify the random walk process. We first construct a process $\{X^{(2)}_t\}_{t \in [T]}$ identical to the process $\{X^{(1)}_t\}_{t \in [\tau]}$, except that the \emph{number of steps} in the random walk increases from $\tau$ to $T$. Let $M^{(2)}_T$ be the number of rounds the random walk step is blocked by the upper boundary $B$. Since $\{X^{(2)}_t\}$ follows the same process but with a larger number of steps, there must be
    \[
    \E[M^+_\tau] ~=~ \E[M^{(1)}_\tau] ~\leq~ \E[M^{(2)}_T].
    \]

    Next, we construct a process $\{X^{(3)}_t\}_{t \in [T]}$ identical to the process $\{X^{(2)}_t\}_{t \in [T]}$, except that the \emph{initial position} in the random walk shifts from $X^{(2)}_0 = B_0$ to $X^{(3)}_0 = B$. Let $M^{(3)}_T$ be the number of rounds the random walk step is blocked by the upper boundary $B$. Since $\{X^{(3)}_t\}$ follows the same process but with an initial position closer to the upper boundary, there must be
    \[
    \E[M^+_\tau] ~\leq~ \E[M^{(2)}_T] ~\leq~ \E[M^{(3)}_T].
    \]

    The next process we construct is $\{X^{(4)}_t\}_{t \in [T]}$, which is identical to the process $\{X^{(3)}_t\}_{t \in [T]}$, except that the \emph{transitioning probability} of the random walk process differs: let $\xi^{(4)}_t$ be the step we aim at taking in round $t$ for process  $\{X^{(4)}_t\}_{t \in [T]}$, and let  $\xi^{(3)}_t$ be the step we aim at taking for process  $\{X^{(3)}_t\}_{t \in [T]}$. Recall that we have 
    \[
    \pr[\xi^{(3)}_t = -1] ~=~ \Expectsell, \quad \text{and} \quad \pr[\xi^{(3)}_t = 1] ~=~ \Expectbuy, \quad \text{and} \quad \pr[\xi^{(3)}_t = 0] ~=~ 1 - \Expectsell - \Expectbuy.
    \]
    For process   $\{X^{(4)}_t\}_{t \in [T]}$, we update the probabilities to
    \[
    \pr[\xi^{(4)}_t = -1] ~=~  \pr[\xi^{(4)}_t = 1] ~=~ \bar \alpha, \quad \text{and} \quad \pr[\xi^{(4)}_t = 0] ~=~ 1 - \Expectsell - \Expectbuy ~=~ \pr[\xi^{(3)}_t = 0],
    \]
    where we recall $\bar \alpha = 0.5 \cdot (\Expectsell + \Expectbuy)$. Since $\Expectsell \geq \Expectbuy$, comparing to process  $\{X^{(3)}_t\}_{t \in [T]}$, process  $\{X^{(4)}_t\}_{t \in [T]}$ keeps the same stay-in-place probability, decreases the probability of going towards $-1$, while increases the probability of going towards $+1$. Let $M^{(4)}_T$ be the number of rounds the random walk step in  process  $\{X^{(4)}_t\}_{t \in [T]}$  is blocked by the upper boundary $B$. Then, there must be
    \[
    \E[M^+_\tau] ~\leq~ \E[M^{(3)}_T] ~\leq~ \E[M^{(4)}_T].
    \]

    Finally, note that the process $\{X^{(4)}_t\}_{t \in [T]}$ and the process $\{\widetilde X_t\}_{t \in [T]}$ exhibit \emph{mirror symmetry}. Therefore, there must be
    \[
     \E[M^+_\tau] ~\leq~ \E[M^{(4)}_T] ~=~ \E[\widetilde M^-_T]. \qedhere
    \]
\end{proof}

To analyze the random walk process $\{\widetilde X_t\}_{t \in [T]}$, we rely on the the following \Cref{lma:balanced-walk-mt-to-xt}, which converts our target from bounding $\E[\widetilde M^-_T]$ to bounding $\E[\widetilde X_T]$.

\begin{Lemma}
    \label{lma:balanced-walk-mt-to-xt}
    For random walk process $\{\widetilde X_t\}_{t \in [T]}$ described in \Cref{lma:balanced-walk}, we have
    \[
    \E[\widetilde M^-_T] ~\leq~ \frac{\bar \alpha T}{B} + \E[\widetilde X_T].
    \]
\end{Lemma}

\begin{proof}
    For simplicity of the proof, we also define $\widetilde M^+_T$ to be the number of rounds the random walk step in  process  $\{\widetilde X_t\}_{t \in [T]}$  is blocked by the upper boundary $B$. By linearity of expectation and the fact that the realization of $\tilde \xi_t$ is independent to the position $\widetilde X_{t-1}$, we have
    \begin{align}
        \E[\widetilde M^-_T] ~=~\sum_{t = 0}^{T - 1} \bar \alpha \cdot \pr[\widetilde X_t = 0], \quad \text{and} \quad \E[\widetilde M^+_T] ~=~\sum_{t = 0}^{T - 1} \bar \alpha \cdot \pr[\widetilde X_t = B] \label{eq:balanced-walk-mt}
    \end{align}

    To prove \Cref{lma:balanced-walk-mt-to-xt}, we first analyze the drift of $\widetilde X_t$. Note that when $\widetilde X_t$ is neither $0$ nor $B$, we have $\E[\widetilde X_{t+1} - \widetilde X_t] = 0$. When $\widetilde X_t = 0$, $\widetilde X_{t+1}$ moves to $1$ with probability $\bar \alpha$ and stays at $0$ otherwise, so $\E[\widetilde X_{t+1} - \widetilde X_t] = \bar \alpha$. Similarly, when  $\widetilde X_t = B$, $\widetilde X_{t+1}$ moves to $B - 1$ with probability $\bar \alpha$ and stays at $B$ otherwise, so $\E[\widetilde X_{t+1} - \widetilde X_t] = -\bar \alpha$. By linearity of expectation, we have 
    \begin{align}
        \E[\widetilde X_T] ~=~\E[\widetilde X_T - \widetilde X_0] ~=~ \bar \alpha \cdot \sum_{t = 0}^T \pr[\widetilde X_t = 0] - \bar \alpha \cdot \sum_{t = 0}^T \pr[\widetilde X_t = B] ~=~ \E[\widetilde M^-_T] - \E[\widetilde M^+_T], \label{eq:balanced-walk-first}
    \end{align}
    where we apply the fact that $\widetilde X_0 = 0$ in the first equality, and  \eqref{eq:balanced-walk-mt} in the last equality.

    Net, we analyze the drift of $\widetilde X^2_t$. Fix $\widetilde X_t$. When $\widetilde X_t$ is not at the boundary, we have
    \[
    \E[\widetilde X^2_{t+1} - \widetilde X^2_t] ~=~ \bar \alpha \cdot \big((\widetilde X_t - 1)^2 - \widetilde X^2_t\big) + \bar \alpha \cdot \big((\widetilde X_t + 1)^2 - \widetilde X^2_t\big) ~=~ 2\bar \alpha.
    \]
    When $\widetilde X_t = 0$, $\widetilde X_{t+1}$ changes to $1$ with probability $\bar \alpha$ and stays otherwise, so $\E[\widetilde X^2_{t+1} - \widetilde X^2_t] = \bar \alpha = 2\bar \alpha - \bar \alpha$. Similarly, when $\widetilde X_t = B$, we have $\E[\widetilde X^2_{t+1} - \widetilde X^2_t] = \bar \alpha \cdot (1 - 2B) = 2\bar \alpha - 2\bar \alpha  B - \bar \alpha$. By linearity of expectation, we have
    \begin{align}
        \E[\widetilde X^2_T] ~=~ \E[\widetilde X^2_T - X^2_0] ~&=~ 2\bar \alpha\cdot T - \alpha \cdot \sum_{t = 0}^{T-1} \pr[\widetilde X_t = 0] -(2\bar \alpha B + \alpha) \cdot \sum_{t = 0}^{T-1} \pr[\widetilde X_t = B] \notag \\
        ~&=~ 2\bar \alpha\cdot T  -  \E[\widetilde M^-_T] -(2B+1) \cdot  \E[\widetilde M^+_T], \label{eq:balanced-walk-second}
    \end{align}
    where the last inequality uses \eqref{eq:balanced-walk-mt}. Combining \eqref{eq:balanced-walk-first} and \eqref{eq:balanced-walk-second} together gives
    \begin{align*}
        \E[\widetilde M^-_T] ~=~ \frac{2 \bar \alpha T + (2B + 1) \cdot \E[\widetilde X_T] - \E[\widetilde X^2_T]}{2B + 2} ~\leq~ \frac{\bar \alpha T}{B} + \E[\widetilde X_T].  
    \end{align*}
\end{proof}

After converting $\E[\widetilde M^-_T]$ into a relatively simple form, we are now ready to prove \Cref{lma:mbuy-bound}.

\begin{proof}[Proof of \Cref{lma:mbuy-bound}]
    By \Cref{lma:balanced-walk}    and \Cref{lma:balanced-walk-mt-to-xt}, it remains to upper-bound the value of  $\E[\widetilde X_T]$ for the balanced random walk process defined in \Cref{lma:balanced-walk}. To achieves this, we introduce the following two random walk processes. Let $\{\widetilde U_t\}_{t \in [T]}$ be the unbounded random walk process w.r.t $\tilde \xi_t$, and $\{\widetilde W_t\}_{t \in [T]}$ be the left-bounded random walk process w.r.t. $\tilde \xi_t$, i.e., we have $\widetilde U_0 = \widetilde W_0 = \widetilde X_0 = 0$, and for $t \in [T]$, we have
\[
\widetilde U_t ~=~ \widetilde U_{t-1} + \tilde \xi_t, \qquad \text{and} \qquad \widetilde W_t = \widetilde W_{t - 1} + \tilde \xi_t \cdot \one\big[W_{t-1} + \xi_t \geq 0\big].
\]
Then, for any realization of $\{\tilde \xi_t\}_{t \in [T]}$, we have
\begin{align}
    \widetilde X_T ~\leq~ \widetilde W_T ~\leq~ \widetilde U_T - \min_{t \in [T]} \widetilde U_T. \label{eq:balanced-walk-xt-to-ut}
\end{align}
We omit the proof of \Cref{eq:balanced-walk-xt-to-ut}, as it is identical to the proof of \Cref{clm:xt-to-ut}. Then, it remains to bound the $\E[\widetilde U_T]$ and $ \E[\min_{t \in [T]} \widetilde U_T]$. We bound these two expectations via Freedman's Inequality. Since $\widetilde U_t$ is an unbounded and balanced random walk process, which means  $\E[\widetilde U_t] = \E[\widetilde U_{t-1}]$ for every $t \in [T]$, the sequence of $\{\widetilde U_t\}_{t \in [T]}$ can be viewed as a martingale process, such that the absolute value of the increments is bounded by
\[
|\Delta \widetilde U_t| ~=~ |\tilde \xi_t| ~\leq~ 1,
\]
and the predictable quadratic variation $\langle \widetilde U \rangle_T$ satisfies
\begin{align*}
    \langle \widetilde U \rangle_T ~=~ \sum_{t \in [T]} \E[\tilde \xi^2_t] ~=~ T \cdot 2 \bar \alpha.
\end{align*}
Let $\epsilon = 6\sqrt{\bar \alpha T} \cdot \log B$. By Freedman's Inequality (\Cref{thm:freedman}), we have
\begin{align*}
    \pr\left[ \max_{t \in [T]} |\widetilde U_t| \geq \epsilon \land \langle \widetilde U \rangle_T \leq 2\bar \alpha T \right] ~\leq~ 2\exp \left(- \frac{\epsilon^2/2}{2\bar \alpha T + 2\epsilon/3}\right) ~\leq~ 2\exp(-3\log B) ~\leq~ B^{-1},
\end{align*}
where the second inequality uses the fact that 
\[
\epsilon^2/2 ~=~ 18 \bar \alpha T \cdot \log^2 B ~\geq~ 6  \bar \alpha T \cdot \log B + 12 \sqrt{\bar \alpha T} \cdot \log^2 B ~=~ 3\log B \cdot (2 \bar \alpha T + 2\epsilon/3),
\]
and the last inequality holds when $B$ is sufficiently large.  Then, we have

\begin{align*}
    \E[\widetilde X_T] ~&\leq~ 2\epsilon \cdot \pr[\widetilde X_T \leq 2\epsilon] + B \cdot \pr\left[\widetilde X_T > 2\epsilon\right] \\
    ~&\leq~ 2\epsilon + B \cdot \pr \left[\widetilde U_T \geq \epsilon \lor - \min_{t \in [T]} \widetilde U_T \geq \epsilon \right] \\
    ~&\leq~ 2\epsilon + B \cdot  \pr\left[ \max_{t \in [T]} |\widetilde U_t| \geq \epsilon \land \langle \widetilde U \rangle_T \leq 2\bar \alpha T \right] \\
    ~&\leq~ 2\epsilon + B \cdot B^{-1} ~=~ 12\sqrt{\bar \alpha T} \cdot \log B + 1 ~\leq~ 13\sqrt{\bar \alpha T} \cdot \log B 
\end{align*}
where the first inequality uses the fact that $\tilde X_T \leq B$ always holds, the second inequality follows from \Cref{eq:balanced-walk-xt-to-ut}, and the third inequality uses the fact that $\langle \widetilde U \rangle_T \leq 2\bar \alpha T$ always holds. Finally, combining the above inequality with \Cref{lma:balanced-walk} and \Cref{lma:balanced-walk-mt-to-xt} gives
\[
\E[M^+_\tau] ~\leq~ \E[\tilde M^-_T] ~\leq~ \frac{\bar \alpha T}{B} + 13\sqrt{\bar \alpha T} \cdot \log B ~\leq~ \frac{\Gamma}{B} + 13\sqrt{\Gamma} \cdot \log B,
\]
where the last inequality uses the fact that $\bar \alpha = 0.5(\Expectsell + \Expectbuy) \leq \Expectsell$.
\end{proof}


